\documentclass[prodmode,permissions]{acmsmall-ec16}
\pdfoutput=1

\usepackage[ruled]{algorithm2e}

\usepackage{amsmath,amsfonts,amssymb,bbm} 
\usepackage[numbers,sort&compress,round]{natbib} 
\SetArgSty{textrm}  
\SetAlFnt{\small}
\SetAlCapFnt{\small}
\SetAlCapNameFnt{\small}
\SetAlCapHSkip{0pt}
\IncMargin{-\parindent}

\doi{XXXXXXX.XXXXXXX}

\DeclareMathAlphabet{\mathitbf}{OML}{cmm}{b}{it}

\newtheorem{fact}[theorem]{Fact}
\newcommand{\qedhere}{}
\renewcommand{\comment}[1]{}
\newenvironment{proofsketch}{\noindent{\sc Proof Sketch:}}{\qed}
\newenvironment{proofof}[1]{\vspace{0.1in}\noindent{\sc Proof of #1.}}{\hfill\qed}

\newcommand{\prob}[2][]{\text{\bf Pr}\ifthenelse{\not\equal{}{#1}}{_{#1}}{}\!\left[#2\right]}
\newcommand{\expect}[2][]{\text{\bf E}\ifthenelse{\not\equal{}{#1}}{_{#1}}{}\!\left[#2\right]}

\newcommand{\yestag}{\addtocounter{equation}{1}\tag{\theequation}}

\newcommand{\agind}[1][i]{_{#1}}

\newcommand{\inversed}[1]{#1^{-1}}
\newcommand{\ironed}{\bar}

\newcommand{\differentiated}[1]{#1'}
\newcommand{\fortype}{\tilde}
\newcommand{\estimated}{\hat}
\newcommand{\sampled}{\hat}

\newcommand{\noaccents}[1]{#1}
\newcommand{\composed}[3]{#1{#2{#3}}}

\newcommand{\newindexedvar}[4][\noaccents]{%
\expandafter\newcommand\expandafter{\csname #2\endcsname}{#1{#4}}%
\expandafter\newcommand\expandafter{\csname #2s\endcsname}{#1{\boldsymbol{#4}}}%
\expandafter\newcommand\expandafter{\csname #2sm#3\endcsname}[1][#3]{#1{\boldsymbol{#4}}_{-##1}}%
\expandafter\newcommand\expandafter{\csname #2#3\endcsname}[1][#3]{#1{#4}\agind[##1]}%
\expandafter\newcommand\expandafter{\csname #2#3th\endcsname}[1][#3]{#1{#4}_{(##1)}}%
}

\newcommand{\rev}{R}

\newcommand{\marg}{R'}
\newcommand{\knalloc}[1][k]{x^{(#1:n)}}
\newcommand{\kalloc}[1][k]{x_{#1}}
\newcommand{\quant}{q}

\newindexedvar{wal}{k}{w}
\newindexedvar[\differentiated]{margwal}{k}{\wal}
\newindexedvar{cumwal}{k}{W}

\newindexedvar{eval}{k}{V} 
\newcommand{\expval}{\mathcal{V}} 
\newcommand{\Zbar}{\bar{Z}} 
\newcommand{\sw}[1][\wals]{S_{#1}} 

\newindexedvar[\ironed]{iwal}{k}{\wal}
\newindexedvar[\ironed]{cumiwal}{k}{\cumwal}
\newindexedvar[\composed{\differentiated}{\ironed}]{margiwal}{k}{\wal}

\newindexedvar{yal}{k}{y}
\newindexedvar{cumyal}{k}{Y}
\newindexedvar[\ironed]{iyal}{k}{\yal}
\newindexedvar[\composed{\differentiated}{\ironed}]{margiyal}{k}{\yal}

\newindexedvar{newyal}{k}{\tilde{y}}
\newindexedvar{cumnewyal}{k}{\tilde{Y}}

\newindexedvar{epswal}{k}{\epsilon}
\newindexedvar{cumepswal}{k}{E}
\newindexedvar[\differentiated]{margepswal}{k}{\epswal}

\newindexedvar{murev}{k}{P}
\newindexedvar[\estimated]{emurev}{k}{\murev}
\newindexedvar[\differentiated]{mumarg}{k}{\murev}

\newindexedvar[\ironed]{imurev}{k}{\murev}
\newindexedvar[\composed{\differentiated}{\ironed}]{imumarg}{k}{\murev}

\newindexedvar{val}{i}{v}

\newindexedvar{eff}{i}{e}

\newindexedvar{bid}{i}{b}
\newindexedvar[\estimated]{ebid}{i}{\bid}
\newindexedvar[\sampled]{sbid}{i}{\bid}

\newindexedvar{strat}{i}{s}

\newindexedvar{virt}{i}{\phi}
\newindexedvar[\inversed]{virtinv}{i}{\virt}

\newindexedvar[\ironed]{ivirt}{i}{\virt}

\newindexedvar{dist}{i}{F}

\newindexedvar{dens}{i}{f}

\newindexedvar{price}{i}{p}
\newindexedvar[\fortype]{tprice}{i}{p}

\newindexedvar{alloc}{i}{x}

\newindexedvar[\fortype]{talloc}{i}{\alloc}

\newindexedvar{util}{i}{u}
\newindexedvar[\fortype]{tutil}{i}{u}

\newcommand{\REV}[1]{P_{#1}}

\newcommand{\samples}{N}
\newcommand{\eps}{\epsilon}

\newcommand{\T}{\mathbf{T}}
\newcommand{\Tij}{\mathbf{T}_{i,j}}
\newcommand{\bwhat}{\widehat{b}}
\newcommand{\btild}{\tilde{b}}
\newcommand{\Phat}{\hat{P}}

\newcommand{\Err}[1]{\expect[\hat{\bid}]{|\hat{#1}-#1|}}

\begin{document}

\title{A/B Testing of Auctions}
\author{SHUCHI CHAWLA
\affil{University of Wisconsin - Madison}
JASON HARTLINE
\affil{Northwestern University}
DENIS NEKIPELOV
\affil{University of Virginia}
}

\begin{abstract}
  For many application areas A/B testing, which partitions users of a
  system into an A (control) and B (treatment) group to experiment
  between several application designs, enables Internet companies to
  optimize their services to the behavioral patterns of their users.
  Unfortunately, the A/B testing framework cannot be applied in a
  straightforward manner to applications like auctions where the users
  (a.k.a., bidders) submit bids before the partitioning into the A and
  B groups is made.  This paper combines auction theoretic modeling
  with the A/B testing framework to develop methodology for A/B
  testing auctions.  The accuracy of our method 
  is directly comparable to ideal A/B
  testing where there is no interference between A and B.  Our results
  are based on an extension and improved analysis of the inference
  method of \citet{CHN-14}.
\end{abstract}



\maketitle


\section{Introduction}
\label{s:intro}

A common method in the practice of large scale auction design, e.g.,
in auctions placing advertisements on online media and Internet search
engines, is A/B testing.  In A/B testing, the auction house is running
an incumbent mechanism A, and would like to determine if a novel
mechanism B obtains higher revenue.  This is done by splitting the
traffic so that most of it goes to A and some of it, e.g., five to ten
percent, goes to B.\footnote{For very thin markets, the ratio of the
  split can even be 50-50.}  An issue with this approach is that if
the bidders are unaware of which mechanism their bid will be
considered in, the bid equilibrium is neither for A nor B but for a
mechanism C that is a convex combination of A and B.

A miscalculation sometimes performed in practice is to consider and
compare average revenue from A (resp.\@ B) from the times when A
(resp.\@ B) is run.  This miscalculation is equivalent to simulating A
on the bids in C and can often give the opposite conclusion; e.g., if
A and B are one- and two-unit highest-bids-win pay-your-bid auctions,
respectively, then B will always appear in this miscalculation to have
higher revenue.  For a fixed set of bids, a pay-your-bid mechanism's
revenue is monotone in its allocation probabilities.  Of course, in
equilibrium, increased allocation probabilities can cause reduced
revenue as bidders may lower their bids.

We present an A/B testing method that applies generally to the position auction
model popularized by the \citet{var-06} and \citet{EOS-07} analyses of
auctions for sponsored search and now a fundamental model for the
study of auction theory; e.g., see \citet{har-13}.  To revisit the
example auction scenario above, notice that the mechanism C, namely
the 50--50 convex combination of one- and two-unit auctions, is
mathematically equivalent to a position auction with weights 1 and 0.5
for the top two slots respectively, and zero for all remaining slots.
Mechanisms A and B are also position auctions.  More generally, a
position auction is defined by a decreasing sequence of weights,
bidders are assigned to positions in decreasing order of bids, and
payments are charged.  Typical payment rules are ``generalized first
price'' and ``generalized second price''; the former requires bidders
to pay their weighted bid, whereas the latter requires bidders to pay the
weighted bid of the next highest bidder. (Unfortunately, for reasons we
describe subsequently, our methods do not apply to the generalized
second price position auction.)  Position auctions with exogenously
given weights are the standard model of auction for selling
advertisements on Internet search engines.  In this application, the
weights correspond to the probability that the ad in each position on
the search results page will be clicked on (and payments are scaled by
weights because bidders only pay if they are clicked on).

\citet{CHN-14} developed an estimator for the revenue of one position
auction (A or B) from the bids of another (C).  This estimator is a
simple weighted order statistic of the distribution of equilibrium bids in C.  To
estimate the revenue of B (analogously for A):
\begin{enumerate}
\item for any number $N$, the evaluation of a formula based on the
  definition of B and C gives $N$ weights, and,
\item given a sorted list of $N$ i.i.d.\@ samples from the bid
  distribution for C, the estimated revenue of B is the weighted average of
  these bids.
\end{enumerate}
This method of estimation compares favorably to {\em ideal A/B
  testing}, where the bids in A are in equilibrium for A and
revenue can be estimated by simple averaging, respectively for B.

Our model of A/B testing is especially applicable to the following
generalization of the classical problem with fixed position weights.
Consider the problem Google and Bing face in selecting the layout of
the sponsored results on the results page of a search query.  The main
variable in the layout is the number of ads to display in the
mainline, i.e., above the organic search results, versus the sidebar,
i.e., to the right of the organic search results.  Typical layouts
include between zero and three advertisements in the mainline, and up
to six on the sidebar.  The weights, which correspond to click
probabilities, depend on which of these four layout choices is
selected. The optimal number of ads to show on the mainline is
distinct for each search query and our methods allow these layouts to
be optimized individually.  Notice that the selection of the number of
mainline ads to show generalizes the opening example of selecting
whether to sell one or two units of an item.

The present paper refines the analysis of the estimator of
\citet{CHN-14}.  The previous paper bounded the error from estimating
the $k$-unit auction revenue from the bids in any position auction
(for any $k$).  As the revenue of any position auction is given by a
convex combination of multi-unit auction revenues, the error from
estimating the revenue of any position auction is bounded by the worst
of these error bounds.  The present paper significantly improves the
error bound of the estimator for the multi-unit auction revenues and
adapts the more sophisticated bounding method to directly bound the
error of an estimator for one position auction from the bids in
another position auction in terms of how close the former is to the
latter.  In particular, mechanism B is close to A/B test mechanism C.

In an A/B test where mechanism C puts $\epsilon$ probability on B (and
$1-\epsilon$ probability on A), our error bounds have better
dependence on $\epsilon$ than ideal A/B testing.  In ideal A/B testing
with a total of $N$ bid samples, there are $\epsilon N$ samples from
B.  Thus, the standard statistical error in estimating the revenue of B is
$O(1/\sqrt{\epsilon N})$ which is an $1/\sqrt{\epsilon}$ dependence on
$\epsilon$.  Our error bounds depend on $\epsilon$ as $\log
1/\epsilon$, a significant improvement.  In particular, our methods
give good bounds even when $\epsilon < 1/N$ and the ideal A/B test is
unlikely to have ever run the B mechanism.  The reason for this
improvement is that, while the ideal A/B test mechanism only sees
$\epsilon N$ bids that depend on B, all bids in the our A/B test
mechanism C depend a little bid on the presence of B.

In terms of the number $N$ of samples from the bid distribution our
estimation methods achieve the optimal rate of $\sqrt N$, i.e., the
mean absolute error is $O(1/\sqrt N)$.  When the A/B test is employed
not to estimate revenues but to determine whether to replace incumbent
auction A with novel auction B, the misclassification error is
exponentially small in $N$.  For our application to advertising on
Internet search engines, deciding the number of ads to show on the
mainline (between zero and four), again the misclassification
probability is exponentially small in $N$.  The precise theorems given
in the paper give the dependence on the number of bidders in the
auction and the fraction of traffic sent to each of the A and B
mechanisms.

\paragraph{Summary of Results}


Our main results apply to first-price and
all-pay position auctions and are:

\begin{enumerate}

\item We improve the analysis of \citet{CHN-14} for the error in
  estimating the revenue of a multi-unit auction from the bids in a
  position auction.  (Section~\ref{s:param-inf} and Appendix~\ref{s:fp-inf}.)

\item We give a direct analysis of the estimation of the revenue of
  B from the bids of C which depends on the measure $\epsilon$ on
  B in the support of C.  The improvement of (1) and this direct
  analysis enables the estimation of the revenue of B from C $=$
  $(1-\epsilon)$A $+$ $\epsilon$B with logarithmic dependence on
  $1/\epsilon$ which allows good error bounds even when $\epsilon$ is
  exponentially small. (Section~\ref{s:ab-revenue}.)

\item 
The logarithmic dependence on $1/\epsilon$ allows for a simple
``universal B test'', a B mechanism that, when mixed with any A
mechanism, allows the revenue of any other position auction to be estimated.
This universal B is a uniform mixture of the 1-unit auction and the
$(n-1)$-unit auction.  (The approach of \citet{CHN-14} was to mix
uniformly over all multi-unit auctions.)

\item If we do not wish to estimate revenues but only determine which
  of A or B will achieve higher revenue, we give an estimator that is
  accurate with probability approaching 1 at an exponential rate in
  $N$, the number of observed bids from the A/B test mechanism C.  In
  comparison, estimation of the auction revenues has quadratic
  rate. (Section~\ref{s:direct-comparison}.)

\item We complement our theoretical analysis with simulation bounds
  that show that our methods are indeed very practical. (Section~\ref{s:simulations}.)

\item We generalize the analysis to allow for the estimation of
  welfare. (Appendix~\ref{s:welfare}.)

\end{enumerate}


\paragraph{Comparison to Global A/B Tests}

Our A/B testing framework is motivated specifically by the goal of
optimizing an auction to local characteristics of the market in which
the auction is run.  It is important to distinguish this goal from
that of another framework for A/B testing which is commonly used to
evaluate global properties of auctions across a collection of disjoint
markets.  This framework randomly partitions the individual markets
into a control group (where auction A is run) and a treatment group
(with auction B).  From such an A/B test we can evaluate whether it is
better for all markets to run A or for all to run B.  It cannot be
used, however, for our motivating application of determining the
number of mainline ads to show, where the optimal number naturally
varies across distinct markets.  The work of \citet{OS-11} on reserve
pricing in Yahoo!'s ad auction demonstrates how such a global A/B test
can be valuable.  They first used a parametric estimator for the value
distribution in each market to determine a good reserve price for that
market.  Then they did a global A/B test to determine whether the
auction with their calculated reserve prices (the B mechanism) has
higher revenue on average than the auction with the original reserve
prices (the A mechanism).  Our methods relate to and can replace the
first step of their analysis.

\paragraph{Revenue versus Click Estimation}

The primary focus of this paper and our discussion thus far is on A/B
testing of auctions when the details of the auction are known.  For
our motivating application of position auctions when we vary the
number of ads shown on the mainline (a.k.a., the configuration), we
assumed that the click-through rates (a.k.a., position weights) of the
positions for each configuration are known.  These weights are
generally not known, and can also be estimated.  Estimation of these
weights is much easier that estimating revenue and ideal A/B testing
gives good estimates.

The distinction between the revenue estimation and click-through rate
estimation problems is the following.  Revenue is determined from bids
that are submitted in advance of the partitioning choice of whether to
run mechanism A or B; consequently, the bids are not in equilibrium for
either A or B but for their convex combination C.  Click-through rates
arise from the behavior of the viewer of the search results page and
this viewer's behavior is affected by the layout directly.  Thus,
while revenue must be estimated via inference methods like ours, the
click-through rates can be estimated by ideal A/B testing.

In the {\em separable model}, where the click-through rates are
exogenous to the layout and, importantly, do not depend on which ads
are shown,\footnote{Though inaccurate, this is one of the most
  standard models under which the search advertising auction is
  studied.} they can be estimated at the same time as we A/B test the
auction formats for revenue.  Our estimators are linear in the
click-through rates; therefore, unbiased estimators for click-through
rates can be used in place of exact click-through rates and the
resulting revenue estimates will be unbiased.  Furthermore, as the
separable model posits that click rates are independent from bids,
estimation errors in click-through rates can be easily factored in to
the error of the revenue estimates.

\paragraph{The Generalized Second Price Auction}

Our analysis does not extend to the generalized second price auction,
i.e., to position auctions where each bidder pays the weighted bid of
the next highest bidder.  The first two issues that arise for the
generalized second price auction are: (1) efficient symmetric
equilibria may not exist \citep{GS-09}; (2) even when a symmetric
equilibrium does exist, it may not be unique \citep{CH-13}.  For these
reasons there is not a one-to-one correspondence between position
weights and the allocation rule of the mechanisms.  

Even restricted to environments where a unique efficient symmetric
Bayes-Nash equilibrium exists, Bayes-Nash inference is challenging
because bidder payments are not a simple function of their bids and
inverting the bid function is problematic. Specifically, for the
generalized first price auction, knowing the bid distribution exactly
allows one to obtain the value distribution via a formula that for any
quantile maps the bid and bid density at that quantile to the value at
that quantile (see equation~\eqref{eq:fp-inf}). On the other hand, for the
generalized second price auction, estimating the value at some
quantile requires knowing the bid and the bid density not only at that
same quantile, but at all quantiles. Importantly, the function mapping
bids to values is not linear in the bids. This makes inference
challenging, and can potentially lead to large errors. Extending our
methodology to this setting is an important open problem.

\paragraph{Related Work}

Our work is motivated, in part, by field work in the past decade that
considers the empirical optimization of reserve prices in auctions
(e.g., \citealp{Ril-06}; \citealp{BM-09}; and \citealp{OS-11}). The
most notable of these is the work of \citet{OS-11}. \citet{OS-11}
adapt their mechanism over time to respond to empirical data by
determining the optimal reserve price for the empirically observed
distribution, and then setting a reserve price that is slightly
smaller. This allows for inference around the optimal reserve price
and ensures that the mechanism quickly adapts to changes in the
distribution.

Prior to our previous work \cite{CHN-14}, to our knowledge, the
problem of inferring parameters of one mechanism using observations
drawn from another has not been considered from a theoretical
perspective. Other lines of work consider optimizing mechanisms based
on samples from the value distribution (see \citealp{CR-14}, and
\citealp{FHHK-14}); on the fly in one mechanism (see
\citealp{GHKSW-06}, \citealp{seg-03}, and \citealp{BV-03}); or by
repeatedly running auctions on agents drawn from the same population,
using multi-armed bandit based approaches (see \citealp{KL-03},
\citealp{BH-05}, and \citealp{CGM-15})). These works exclusively
consider mechanisms that have truthtelling equilibria and for which,
consequently, inference is trivial.


Within econometrics literature, several recent works have developed
techniques for estimating the distribution of values from the
distribution of bids of a given auction (see, e.g., \citet{guerre},
\citet{paarsch}, and \citet{athey:2007}). However, estimating
parameters of one mechanism from the bids of another has not been
considered. While these inference approaches form the basis over which
we develop our methodology, our work shows that for certain quantities
of interest such as the multi-unit auction revenues or the ranking of
revenues of different mechanisms, we can bypass learning the
underlying value distribution altogether, leading to better
convergence bounds.

\section{Preliminaries}\label{s:prelim}

%
%
%
%
%


A standard auction design problem is defined by a set $[n] =
\{1,\ldots,n\}$ of $n\ge 2$ agents, each with a private value $\vali$
for receiving a service.  The values are bounded: $\vali\in[0,1]$;
They are independently and identically distributed according to a
continuous distribution $\dist$.   An
auction elicits bids $\bids = (\bidi[1],\ldots,\bidi[n])$ from the
agents and maps the vector $\bids$ of bids to an allocation
$\tallocs(\bids) = (\talloci[1](\bids),\ldots,\talloci[n](\bids))$,
specifying the probability with which each agent is served, and prices
$\tprices(\bids) = (\talloci[1](\bids),\ldots,\talloci[n](\bids))$,
specifying the expected amount that each agent is required to pay.  An
auction is denoted by $(\tallocs,\tprices)$.

\paragraph{Standard payment formats}

In this paper we study two standard payment formats. In a {\em
  first-price} format, each agent pays his bid upon winning, that is,
$\tpricei(\bids) = \bidi \, \talloci(\bids)$. In an {\em all-pay} format,
each agent pays his bid regardless of whether or not he wins, that is,
$\tpricei(\bids) = \bidi$.

\paragraph{Bayes-Nash equilibrium}
The value distribution $\dist$ is common
knowledge to the agents.  A strategy $\strati$ for agent $i$ is a
function that maps the
value of the agent to a bid.  The distribution $\dist$ and a
profile of strategies $\strats = (\strati[1],\cdots,\strati[n])$
induces interim allocation and payment rules (as a function of bids)
as follows for agent $i$ with bid $\bidi$.
\begin{align*}
\talloci(\bidi) = \expect[\valsmi \sim
\dist]{\talloci(\bidi,\stratsmi(\valsmi))} & \text{ and }\;
\tpricei(\bidi) = \expect[\valsmi \sim \dist]{\tpricei(\bidi,\stratsmi(\valsmi))}.
\intertext{Agents have linear utility which can be expressed in the interm as:}
\tutili(\vali,\bidi) &= \vali\talloci(\bidi) - \tpricei(\bidi).
\end{align*}
The strategy profile $\strats$ forms a {\em Bayes-Nash equilibrium} (BNE) if for
all agents $i$, values $\vali$, and alternative bids $\bidi$, bidding
$\strati(\vali)$ according to the strategy profile is at least as good
as bidding $\bidi$.  I.e.,
\begin{align}
\label{eq:br}
\tutili(\vali,\strati(\vali)) &\geq \tutili(\vali,\bidi).
\end{align}

A symmetric equilibrium is one where all agents bid by the same
strategy, i.e., $\strats$ statisfies $\strati = \strat$ for some
$\strat$.  For a symmetric equilibrium, the interim allocation and
payment rules are also symmetric, i.e., $\talloci = \talloc$ and
$\tpricei = \tprice$ for all $i$.  For implicit distribution $\dist$ and
symmetric equilibrium given by strategy $\strat$, a mechanism can be
described by the pair $(\talloc,\tprice)$.  

The strategy profile allows the mechanism's outcome rules to be
expressed in terms of the agents' values instead of their bids; the
distribution of values allows them to be expressed in terms of the
agents' values relative to the distribution.  This latter
representation exposes the geometry of the mechanism.  Define the {\em
  quantile} $\quant$ of an agent with value $\val$ to be the
probability that $\val$ is larger than a random draw from the
distribution $\dist$, i.e., $\quant=\dist(\val)$.  Denote the agent's
value as a function of quantile as $\val(\quant) =
\dist^{-1}(\quant)$, and his bid as a function of quantile as
$\bid(\quant) = \strat(\val(\quant))$.  The outcome rule of the
mechanism in quantile space is the pair
$(\alloc(\quant),\price(\quant)) =
(\talloc(\bid(\quant)),\tprice(\bid(\quant)))$.

\paragraph{Revenue curves and auction revenue}

\citet{mye-81} characterized Bayes-Nash equilibria and this
characteriation enables writing the revenue of a mechanism as a
weighted sum of revenues of single-agent posted pricings.  Formally,
the {\em revenue curve} $\rev(\quant)$ for a given value distribution
specifies the revenue of the single-agent mechanism that serves an
agent with value drawn from that distribution if and only if the
agent's quantile exceeds $\quant$: $\rev(\quant) =
\val(\quant)\,(1-\quant)$. $\rev(0)$ and $\rev(1)$ are defined as
$0$. Myerson's characterization of BNE then implies that the expected
revenue of a mechanism at BNE from an agent facing an allocation rule
$\alloc(\quant)$ can be written as follows:
\begin{eqnarray}
  \label{eq:bne-rev}
  \REV{\alloc} = \expect[\quant]{\rev(\quant)\alloc'(\quant)} = -\expect[\quant]{\rev'(\quant)\alloc(\quant)} 
\end{eqnarray}
where $\alloc'$ and  $\rev'$ denote the derivative of $\alloc$ and $\rev$ with respect to $\quant$, respectively.

The expected revenue of an auction is the sum over the agents of its
per-agent expected revenue; for auctions with a symmetric equilibrium
allocation rule $\alloc$ this revenue is $n \, \REV{\alloc}$.

\paragraph{Position environments}
\label{sec:rank-based-basics}

A {\em position environment} expresses the feasibility constraint of
the auction designer in terms of {\em position weights} $\wals$
satisfying $1 \ge \walk[1]\ge \walk[2] \ge \cdots \ge \walk[n] \geq
0$.  A {\em position auction} assigns agents (potentially randomly) to
positions $1$ through $n$, and an agent assigned to position $i$ gets
allocated with probability $\walk[i]$.  The {\em rank-by-bid position
  auction} orders the agents by their bids, with
ties broken randomly, and assigns agent $i$, with the $i$th largest
bid, to position $i$, with allocation probability $\walk[i]$.  {\em
  Multi-unit environments} are a special case and are defined for $k$
units as $\walk[j] = 1$ for $j \in \{1,\ldots,k\}$ and $\walk[j] = 0$
for $j \in \{k+1,\ldots,n\}$.  The {\em highest-$k$-bids-win}
multi-unit auction is the special case of the rank-by-bid position
auction for the $k$-unit environment.

Rank-by-bid position auctions can be viewed as convex combinations of
highest-bids-win multi-unit auctions.  The {\em marginal weights} of a
position environment are $\margwals =
(\margwalk[1],\ldots,\margwalk[n])$ with $\margwalk = \walk -
\walk[k+1]$.  Define $\margwalk[0] = 1-\walk[1]$ and note that the
marginal weights $\margwals$ can be interpreted as a probability
distribution over $\{0,\ldots,n\}$.  The rank-by-bid position auction
with weights $\wals$ has the exact same allocation rule as the
mechanism that draws a number of units $k$ from the distribution given
by $\margwals$ and runs the highest-$k$-bids-win auction.

In our model with agent values drawn i.i.d.\@ from a continuous
distribution, rank-by-bid position auctions with either all-pay or
first-price payment semantics have a unique Bayes-Nash equilibrium and
this equilibrium is symmetric and efficient, i.e., in equilibrium, the
agents' bids and values are in the same order \citep{CH-13}.  Denote
the highest-$k$-bids-win allocation rule as $\knalloc$ and its revenue
as $\murevk = \REV{\knalloc} =
\expect[\quant]{-\marg(\quant)\,\knalloc(\quant)}$.  This allocation
rule is precisely the probability an agent with quantile $\quant$ has
one of the highest $k$ quantiles of $n$ agents, or at most $k-1$ of
the $n-1$ remaining agents have quantiles greater than $\quant$.
Formulaically,
\begin{align*}
\knalloc(\quant) &= \sum_{i=0}^{k-1} \tbinom{n-1}{i}
\quant^{n-1-i}(1-\quant)^{i}.  
\\\intertext{Importantly, the allocation rule of a rank-by-bid
  position auction does not depend on the distribution at all.  The
  allocation rule $\alloc$ of the rank-by-bid position auction with
  weights $\wals$ is:}
\alloc(\quant) &= \sum\nolimits_k \margwalk\, \knalloc(\quant).
\\\intertext{By revenue equivalence \citep{mye-81}, the per-agent
  revenue of the rank-by-bid position auction with weights $\wals$
  is:}
\REV{\alloc} &= \sum\nolimits_k \margwalk\,\murevk.
\end{align*}
Of course, $\murevk[0] = \murevk[n] = 0$ as always serving or never
serving the agents gives zero revenue.

\paragraph{Inference}

The distribution of values, which is unobserved, can be inferred from
the distribution of bids, which is observed.  This derivation is
summarized by Lemma~\ref{l:inference}, below.  Interested readers can
find a detailed derivation in \citet{CHN-14} or
Appendix~\ref{app:inference}.

\begin{lemma}
\label{l:inference}
For an auction with first-price payment format, the value function
$\val(\cdot)$ satisfies,
\begin{align}
  \label{eq:fp-inf}
  \val(\quant) &= \bid(\quant) + \tfrac{\alloc(\quant)\,\bid'(\quant)}{\alloc'(\quant)}.\\
\intertext{For an auction with all-pay payment format, it satisfies,}
\label{eq:ap-inf}
  \val(\quant) & = \tfrac{\bid'(\quant)}{\alloc'(\quant)}.
\end{align}
\end{lemma}

Recall that the functions $\alloc(\quant)$ and $\alloc'(\quant)$ are
known precisely for position auctions, while functions $\bid(\quant)$
and $\bid'(\quant)$ are observed in the data and can be estimated.

\paragraph{Statistical Methods}

The errors in the estimated bid distribution follow from standard
analyses.  For $\samples$ draws from the bid distribution, define the
{\em estimated bid distribution} $\ebid(\cdot)$ as $\ebid(\quant) =
\sbidi$ for all $i \in \samples$ and $\quant \in [i-1,i)/\samples$,
where $\sbidi$ is the $i$th smallest bid in the sample.  For function
$\bid(\cdot)$ and estimator $\ebid(\cdot)$, the {\em uniform mean absolute
  error} as a function of the number of samples $\samples$
is $$\expect{ \sup\nolimits_\quant \big| \bid(\quant) -
\ebid(\quant) \big|}.$$

The main object that will arise in our further analysis will be the weighted
quantile function of the bid distribution where the weights are determined
by the allocation rule of the auction under consideration. Our statistical
results stem from the previous work on the uniform convergence of 
quantile processes and weighted quantile processes in \citet{csorgo:78}, \citet{csorgo:83}, \citet{cheng:97}.
It turns out that the our main object of interest for inference is
the weighted quantile function of the bid distribution. The weight is proportional
to the inverse derivative of the allocation rule. This feature leads to highly desirable
proprties of the $\sqrt{N}$-normalized mean absolute error, making it bounded by a universal constant.

\begin{lemma}
\label{error bid function}
The uniform mean absolute error of the empirical quantile function $\ebid(\cdot)$ 
weighted by its derivative is bounded as follows, as $N \rightarrow \infty$.
$$\expect{
  \sup\nolimits_\quant \big|\sqrt{N}(b^{\prime}(q))^{-1}( \bid(\quant) - \ebid(\quant)) \big|} <
1.$$ 
\end{lemma}

This lemma along with Equation~\eqref{eq:fp-inf}, and using the fact
that values are bounded by $1$ gives the following bound on the 
uniform mean absolute error of the weighted bid distribution for the first-price auction as
$$ \textstyle \expect{\sup\nolimits_\quant \big|\sqrt{N}\frac{x(q)}{x'(q)}\left(
  \bid(\quant) - \ebid(\quant) \right)\big|}  \leq
1.$$
Likewise, using Equation~\eqref{eq:ap-inf} for the all-pay auction, we
get:
$$ \textstyle \expect{\sup\nolimits_\quant \big|\sqrt{N}(x^{\prime}(q))^{-1}
\left(
  \bid(\quant) - \ebid(\quant)\right) \big| } \leq 
1.$$

Equations~\eqref{eq:fp-inf} and~\eqref{eq:ap-inf} enable the value
function, or equivalently, the value distribution, to be estimated
from the estimated bid function and an estimator for the derivative of
the bid function, or equivalently, the density of the bid
distribution. Estimation of densities is standard; however, it
requires assumptions on the distribution, e.g., continuity, and the
convergence rates in most cases will be slower.  Our main results do not take
this standard approach. 

\section{Inference methodology and error bounds for all-pay auctions}
\label{s:param-inf}

We will now develop a methodology and error bounds for estimating the
revenue of one rank-based auction using bids from another rank-based
auction. There are two reasons behind our assumption that the auction that we
run (that generates the observed bids) is a rank-based auction. First,
the allocation rule (in quantile space) of a rank based auction is
independent of the bid and value distribution; therefore, it is known
and does not need to be estimated.  Second, the allocation rules that
result from rank-based auctions are well behaved, in particular their
slopes are bounded, and our error analysis makes use of this property.

Recall from Section~\ref{sec:rank-based-basics} that the
revenue of any rank-based auction can be expressed as a linear
combination of the multi-unit revenues $\murevk[1],\ldots,\murevk[n]$
with $\murevk$ equal to the per-agent revenue of the
highest-$k$-bids-win auction. Therefore, in order to estimate the
revenue of a rank-based auction, it suffices to estimate the
$\murevk$s accurately.

In the following section we describe a function mapping the observed
bids to the revenue estimate. In Sections~\ref{s:inference-k} and
\ref{s:inference-y} we develop error bounds on our estimate for the
revenue of a multi-unit auction and a general position auction respectively.



\subsection{Inference equation}
\label{s:inference-equation}

Consider estimating the revenue of an auction with allocation rule $y$
from the bids of an all-pay position auction.  The per-agent revenue
of the allocation rule $y$ is given by:
$$
P_y = \expect[\quant]{y'(\quant)\rev(\quant)} =
\expect[\quant]{y'(\quant)\val(\quant) (1-\quant)}
$$


Let $\alloc$ denote the allocation rule of the auction that we run,
and $\bid$ denote the bid distribution in BNE of this auction.  Recall
that for an all-pay auction format, we can convert the bid
distribution into the value distribution as follows: $\val(\quant) =
\bid'(\quant)/\alloc'(\quant)$.  Substituting this into the expression
for $P_y$ above we get
$$
P_y =
\expect[\quant]{y'(\quant)(1-\quant)\frac{\bid'(\quant)}{\alloc'(\quant)}}
= \expect[\quant]{Z_y(\quant)\bid'(\quant)},\; \text{ where }\, Z_y(\quant)=(1-\quant)\frac{y'(\quant)}{\alloc'(\quant)}.
$$

This expression allows us to derive the revenue using the empirical
bid distribution.  
In \citet{CHN-14} we observe that integration by
parts yields an estimator that is a simple weighted average of the
empirical bid distribution.  

\begin{lemma}[\citet{CHN-14}] \label{weights lemma}
  The per-agent revenue of a rank-based auction with allocation rule $y$ can be written as
  a linear combination of the bids in an all-pay auction:\footnote{Note that $\bid(0)=0$ and $Z_k(1)=0$.}
$$
P_y = \expect[\quant]{-Z'_y(\quant)\bid(\quant)}
$$
where $Z_y(\quant)
=(1-\quant)\frac{y(\quant)}{\alloc'(\quant)}$ depends on the
allocation rule $\alloc$ of the mechanism and is known precisely.
\end{lemma}

Recall that the estimated bid
distribution $\hat{b}(\cdot)$ is defined
as a piecewise constant function with $N$ pieces.  Thus,
the estimator $\hat{P}_y = \expect[q]{\smash{-Z'_y(q)\,\hat{\bid}(q)}}$ can be
simplified as expressed in the following definition.

\begin{definition}
\label{d:estimator}
The estimator $\hat{P}_y$ for the revenue of an auction with
allocation rule $y$ from $N$ samples $\hat{b}_1 \leq \cdots \leq
\hat{b}_N$ from the equilibrium bid distribution of an auction with
allocation rule $x$ is:
\begin{align}
\label{AB-test-estimator-sum} 
\hat{P}_y& = \sum_{i=1}^{N} \left[ \left(1-\frac {i-1} N\right)\frac{y'(\frac {i-1}N)}{x'( \frac {i-1}N)} - \left(1-\frac{i}{N}\right)\frac{y'(\frac{i}N)}{x'(\frac{i}N)} \right] \, \hat{\bid}_i
\end{align}
\end{definition}



\subsection{Error bound for the multi-unit revenues}
\label{s:inference-k}

We will now develop an error bound for the estimator for
Lemma~\ref{weights lemma} for the case of the multi-unit revenues.  In
the following, denote the the allocation rule of the
highest-$k$-bids-win auction as $\kalloc$ (for an implicit number $n$
of agents).  We can therefore express the error in the estimation of
$P_k$ in terms of the error in estimating the bid distribution. 
\begin{eqnarray}
|\emurevk-\murevk| & = &
\left|\expect[q]{-Z'_k(\quant)(\widehat{\bid}(\quant)-\bid(\quant))}\right|
\label{eq:error-allpay}
\end{eqnarray}

In \citet{CHN-14} we gave a simple but weak analysis of the error,
obtaining the following bound:
$$\Err{\murevk}\le \sqrt{\frac{2}{\samples}}\,\, \sup_{\quant}\{\alloc'(\quant)\}\,\,\sup_{\quant}
\left\{ \frac{\kalloc'(\quant)}{\alloc'(\quant)} \right\} $$

To understand this error bound, we note that the maximum slope of the
multi-unit allocation rules $\kalloc$, and therefore also that of any
rank-based auction, is always bounded by $n$, the number of agents in
the auction (summarized as Fact~\ref{fact:max-slope} in Appendix~\ref{s:app-proofs}). This
bounds the $\sup_{\quant} \alloc'(\quant)$ term.
The quantity $\sup_{\quant} \left\{
  \frac{\kalloc'(\quant)}{\alloc'(\quant)} \right\}$ on the other hand
may be rather large, even unbounded, if $\alloc'(\cdot)$ is near zero
at some $\quant$.  The approach taken in \citet{CHN-14} is to
explicitly design $\alloc$ with position weight $\margwalk > \epsilon$
(i.e., to mix $\kalloc$ into $\alloc$) to ensure this latter term is
bounded. We will now develop a stronger bound on the error in
$\murevk$ with better dependence on the quantity $\sup_{\quant}
\left\{ \frac{\kalloc'(\quant)}{\alloc'(\quant)} \right\}$.

We start with Equation~\eqref{eq:error-allpay}
and partition the expectation on the right hand side of the equation
as follows for some $\alpha>0$:
\begin{eqnarray*}
|\emurevk-\murevk| & = &
\left|\expect[q]{-Z'_k(\quant)(\widehat{\bid}(\quant)-\bid(\quant))}\right|\\
& \le & \expect[q]{\frac{\left(\log(1+Z_k(q))\right)^{\alpha}}{Z_k(q)}|Z'_k(q)|}
\sup\limits_q\left|\frac{Z_k(q)}{\left(\log(1+Z_k(q))\right)^{\alpha}}(\hat{b}(q)-b(q))\right|
\end{eqnarray*}
An appropriate choice of $\alpha$ gives us the following theorem. We
defer the proof to the appendix.
\begin{theorem}\label{th: all pay}
  Let $\kalloc$ denote the allocation function of the
  $k$-highest-bids-win auction and $\alloc$ be the allocation function
  of any rank-based auction. Then for all $k$, the mean squared error
  in estimating $P_k$ from $\samples$ samples from the bid
  distribution for an all-pay auction with allocation rule $\alloc$ is
  bounded by:
$$
\Err{\murevk}\le \frac{40}{\sqrt{\samples}}\,\,
\sup_{\quant}\{\kalloc'(\quant)\}\,\,\log 
\max\left\{\sup_{\quant: \kalloc'(\quant)\ge 1}\frac{\alloc'(\quant)}{\kalloc'(\quant)}, \sup_{\quant}\frac{\kalloc'(\quant)}{\alloc'(\quant)} \right\}.
$$
\end{theorem}

\noindent
Invoking Fact~\ref{fact:max-slope}, we note that the error in
$\murevk$ given by Theorem~\ref{th: all pay} can be bounded by
$$\Err{\murevk}= O\left(\frac{n}{\sqrt{\min\{k, n-k\}}}\right) \frac{1}{\sqrt{\samples}} \log 
\max\left\{\sup_{\quant: \kalloc'(\quant)\ge 1}\frac{\alloc'(\quant)}{\kalloc'(\quant)}, \sup_{\quant}\frac{\kalloc'(\quant)}{\alloc'(\quant)} \right\}.
$$

\noindent To understand the above error bound better, we make the following
observations:
\begin{itemize}
\item When $\alloc=\kalloc$, we get an error bound of $O(\sup_q
  \kalloc'(q)/\sqrt{N})$, which is the same (within constant factors)
  as the statistical error in bids.
\item We also get a good error bound when $\alloc$ and $\kalloc$ are
  close enough without being identical: when $\eps\kalloc'\le
  \alloc'\le \kalloc'/\eps$, we get a bound of $O(\log(1/\eps) \sup_q
  \kalloc'(q)/\sqrt{N})$.
\item Finally, as long as $\alloc'\ge \eps\kalloc'$, that is, the
  highest-$k$-bids-win auction is mixed in with $\epsilon$ probability
  into $\alloc$, we observe via Fact~\ref{fact:max-slope} that $\sup_{q:
    \kalloc'(q)\ge 1} \alloc'(q)/\kalloc'(q) \le \sup_q \alloc'(q)\le
  n$, and obtain an error bound of $O(\log(n/\eps) \sup_q
  \kalloc'(q)/\sqrt{N})$.
\end{itemize}

\subsection{Error bound for arbitrary rank-based revenues}
\label{s:inference-y}

We now develop an error bound for our estimator for the revenue,
$P_y$, of an arbitrary position auction with allocation rule $y$ from
the bids of another position auction $x$. Let us write $y$ as a
position auction with weights $\wals$:
\begin{align*}
y  = \sum\nolimits_k \margwalk\,\kalloc \;\; \text{ and }
P_y & = \sum\nolimits_k \margwalk\,P_k\\
\intertext{Accordingly, the error in $P_y$ is bounded by a weighted sum of the
  error in $P_k$:}
\Err{P_y}  & \le \sum\nolimits_k \margwalk\,\Err{P_k}\\
\intertext{applying Theorem~\ref{th: all pay},}
& \le \frac{40}{\sqrt{N}} \sum\nolimits_k \margwalk \sup_q
\{\kalloc'(q)\} \left(\log n + \log\frac{1}{\margwalk}+ \log\sup_q \frac{y'(q)}{x'(q)}\right)\\
& \yestag \label{eq:weak-bound-py}
= O\left(\frac{n\log n}{\sqrt{N}} \log\sup_q \frac{y'(q)}{x'(q)}\right)
\end{align*}
Unfortunately, the above bound can be quite loose, as the following
simple example demonstrates. Suppose that $x=y$ and $\margwalk=1/n$
for all $k$. Then the above approach (via a tighter bound on the sum
over $k$) leads to an error bound of $O(\sqrt{n}\log n/\sqrt{N})$,
whereas, the true error bound should be $O(1/\sqrt{N})$, arising due
to the statistical error in bids. Furthermore, it is desirable to
obtain an error bound that depends directly on $\sup_q y'(q)$, rather
than on the constituent $\sup_q \kalloc'(q)$; the latter can be much
larger than the former. Below, we analyze the error in $P_y$ directly,
leading to a slightly tighter bound.

\begin{theorem}
\label{thm:inference-y}
  The expected absolute error in estimating the revenue of a position
  auction with allocation rule $y$ using $N$ samples from the bid
  distribution for an all-pay position auction with allocation rule
  $x$ is bounded as below; here $n$ is the number of positions
  in the two position auctions.
\begin{align*}
\Err{P_y}
& \le \frac{40}{\sqrt{N}} \sqrt{n\log n} \,\,\sup_{\quant}\{y'(\quant)\}\,\,\log \max\left\{\sup_{\quant: y'(\quant)\ge 1}\frac{\alloc'(\quant)}{y'(\quant)},
\sup_{\quant}\frac{y'(\quant)}{\alloc'(\quant)} \right\} \\
& + \frac{O(1)}{N} \sup_{\quant}\{\alloc'(\quant)\}\,\,\sup_{\quant}
\left\{ \frac{y'(\quant)}{\alloc'(\quant)} \right\}
\end{align*}
\end{theorem}

Note that the first term in the error bound in
Theorem~\ref{thm:inference-y} dominates, and this term is identical to
the bound in Theorem~\ref{th: all pay}, except for an extra
$\sqrt{n\log n}$ term. Moreover, when $\sup_q y'(q)<\sqrt{n}$,
Theorem~\ref{thm:inference-y} gives us a tighter error bound than
Equation~\eqref{eq:weak-bound-py}.  We will now prove the
theorem. Here we provide an outline for the proof; the complete
argument can be found in the appendix.

\vspace{0.1in}
\begin{proofsketch}
As for the multi-unit revenues,
\begin{align}
\notag |\Phat_y-P_y| & = 
\left|\expect[q]{Z_y(\quant)(\bwhat'(\quant)-\bid'(\quant))}\right| = 
\left|\expect[q]{Z'_y(\quant)(\bwhat(\quant)-\bid(\quant))}\right|
\end{align}


In order to simplify our analysis of the error in $P_y$, we will break
up the error into two components: the bias in the estimator
$\hat{P}_y$ and the deviation of $\hat{P}_y$ from its mean.
\begin{align}
\notag |\Phat_y-P_y| & \le \left|\Phat_y - \expect{\Phat_y}\right| + \left|\expect{\Phat_y}-P_y\right|\\
& = \left|\expect[q]{Z'_y(\quant)(\bwhat(\quant)-\btild(\quant))}\right|
+ \left|\expect[q]{Z'_y(\quant)(\btild(\quant)-\bid(\quant))}\right|
\label{eq:error-allpay-y-hat}
\end{align}
Here, $\btild$ is a step function that equals the expectation of the
empirical bid function $\bwhat$:
\begin{align*}
\btild(\quant) & = \expect{\bwhat(\quant)}  
\end{align*}
The bias of the estimator $\Phat$ (i.e. the second term in
\eqref{eq:error-allpay-y-hat}) is easy to bound. We prove in the
appendix that $\btild(\quant)-\bid(\quant)$ is at most $O(1)/N$ times $\sup_q\{
  \alloc'(q)\}$. This implies the following lemma.
\begin{lemma}
\label{lem:bias-bound}
  With $\btild$ defined as above,
  $$\left|\expect{\Phat_y}-P_y\right|=\left|\expect[q]{Z'_y(\quant)(\btild(\quant)-\bid(\quant))}\right|=
  \frac{O(1)}{N} \sup_{\quant}\{\alloc'(\quant)\}\,\,\sup_{\quant}
\left\{ \frac{y'(\quant)}{\alloc'(\quant)} \right\}.$$
\end{lemma}

We now focus on the first term in \eqref{eq:error-allpay-y-hat},
namely the integral over the quantile axis of
$Z'_y(\quant)(\bwhat(\quant)-\btild(\quant))$. The approach of
Section~\ref{s:inference-k} does not provide a good upper bound on
this quantity, because a counterpart of Lemma~\ref{lem:Z-bound-1} (see
Appendix~\ref{s:app-proofs}) fails to hold for $Z_y$. Instead, we will
express the integral as a sum over several independent terms, and show
that it is small in expectation.

To this end, we first identify the set of quantiles at which the
function $\bwhat$ ``crosses'' the function $\btild$ from below. This
set is defined inductively. 
Define $i_0=0$. Then, inductively, let
$i_\ell$ be the smallest integer strictly greater than $i_{\ell-1}$
such that
$$\bwhat\left(\frac{i_{\ell}-1}{N}\right)\le\btild\left(\frac{i_\ell-1}{N}\right) \, \text{ and }\,
\bwhat\left(\frac{i_{\ell}}{N}\right)>\btild\left(\frac{i_\ell}{N}\right).$$
Let $i_{k-1}$ be the last integer so defined, and let $i_k=N$. Let $I$
denote the set of indices $\{i_0,\cdots, i_k\}$.
Let $\Tij$ denote the following integral:
$$\Tij = \int_{\quant=i/N}^{\quant=j/N}
Z'_y(\quant)(\bwhat(\quant)-\btild(\quant)) \, \text{d}q$$
Then, our goal is to bound the quantity
$\expect[\hat{\bid}]{|\T_{0,N}|}$ where $\T_{0,N}$ can be written
as the sum:
$$ \T_{0,N} = \sum_{\ell=0}^{\ell=k-1}\T_{i_{\ell}, i_{\ell+1}}. $$
We now claim that conditioned on $I$ and the maximum bid error, this
is a sum over independent random variables. In the following, let $G$
denote the bid distribution, and $\hat{G}$ the empirical bid
distribution.

\begin{lemma}
\label{lem:independent-sums}
Conditioned on the set of indices $I$ and
$\Delta=\sup_\quant|\hat{G}(\bid(\quant))-G(\bid(\quant))|$, over the
randomness in the bid sample, the random variables $\T_{i_{\ell},
  i_{\ell+1}}$ are mutually independent.
\end{lemma}
Then we apply Chernoff-Hoeffding bounds, coupled with the approach
from Section~\ref{s:inference-k} to bound each individual
$\T_{i_{\ell}, i_{\ell+1}}$, to obtain a bound on the proability that
$\expect[\hat{\bid}]{|\T_{0,N}|\, | \, I, \Delta}$ exceeds some value $a>0$.
\begin{lemma}
\label{lem:deviation-prob}
  With $I=\{i_0,\cdots, i_k\}$ and $T_{i,j}$ defined as above, for any $a>0$,
  $$\Pr\left[|\T_{0,N}|\ge a\, | \, I, \Delta\right]\le
  \text{exp}\left(-\frac{a^2}{n(40\Delta C)^2}\right),\,\,\,
  \text{where } C = \sup_{\quant}\{y'(\quant)\}\log\max\left\{\sup_{\quant: y'(\quant)\ge 1}\frac{\alloc'(\quant)}{y'(\quant)},
\sup_{\quant}\frac{y'(\quant)}{\alloc'(\quant)} \right\}.$$
\end{lemma}
Combining this lemma with an absolute bound on $|\T_{0,N}|$, and
removing the conditioning on $I$ and $\Delta$, we obtain the
Theorem~\ref{thm:inference-y}.
\end{proofsketch}

\section{Simulation evidence}

\label{s:simulations}

We now present some simulation evidence to support our theoretical
results. Our focus will be the inference of the revenue of a position
auction B using bid data from mixed auction C $= (1-\epsilon)\,$A$ +
\epsilon\, $B where A is another position auction. Recall that $x_A$
denotes the allocation rule of A, $x_B$ the allocation rule of B,
and $x_C(q) = (1-\epsilon)\,x_A(q) + \epsilon\,x_B(q)$ is the
allocation rule of C.

We consider the following three designs.
\begin{itemize}
\item {\bf Design 1:} Auction A is the one-unit auction, $x_A(q) = q^{n-1}$; and auction B is the uniform-stair position auction, $x_B(q)=q$.
\item {\bf Design 2:} Auction A is the uniform-stair position auction, $x_A(q) = q$; and Auction B is the one-unit auction, $x_B(q)=q^{n-1}$ 
\item {\bf Design 3:} Auction A is the $(n-1)$-unit auction, $x_A(q) = 1-(1-q)^{n-1}$; and auction B is the one-unit auction, $x_B(q)=q^{n-1}$.
\end{itemize}
The uniform-stair position auction is given by position weights
$\wals$ with $\walk = \frac{n-k}{n-1}$ for each $k$.  For each design,
the values of the bidders are drawn the beta distribution with
parameters $\alpha=\beta=2$.  This distribution of values is supported
on $[0,1]$; it is unimodal with the mode and the mean at $1/2$ and it
is symmetric about the mean.

\subsection*{Methodology}

We perform simulations to calculate the mean absolute deviation of our
estimator $\hat{P}_B$ for the revenue of auction B with its expected
revenue $P_B$.  The allocation rules $\alloc_B$ and $\alloc_C$, their
derivatives $\alloc'_B$ and $\alloc'_C$, and the revenue curve $\rev$
are calculated analytically.  The expected revenue $P_B$ is calculated
from the revenue curve $\rev$ and $\alloc'_B$ by
equation~\eqref{eq:bne-rev} via numerical integration (i.e., by
averaging the values of $\rev(\quant)\,\alloc'_B(\quant)$ on a grid).
The equilibrium bids in auction C for values on a uniform grid are
calculated from equation~\eqref{eq:ap-inf} via numerical integration
on a grid.  Each simulation draws $N$ bids from this set of bids with
replacement, the estimated revenue $\hat{P}_B$ is calculated from
Definition~\ref{d:estimator}, and the mean absolute deviation is
calculated by averaging $|P_B - \hat{P}_B|$ over 1000 Monte Carlo
simulations.

\subsection*{Results and observations}
Theorems~\ref{th: all pay} and \ref{thm:inference-y} give the following upper bounds on the
mean absolute deviation of our revenue estimator for the three designs
considered above:
\begin{itemize}
\item {\bf Design 1:} $\frac{40\sqrt{n\log n}}{\sqrt{N}}
\log\max\left(n,\frac 1\eps\right) + O\left(\frac nN\right)$.
\item {\bf Design 2:} $\frac{40n}{\sqrt{N}} \log \frac1\eps$.
\item {\bf Design 3:} $\frac{40n}{\sqrt{N}} \log\frac 1\eps$.
\end{itemize}
In Figure~\ref{table:MC} we report our empirically observed mean
absolute error in revenue for each of the three designs; the auction
mixture is set as $\epsilon = .001$ and parameters $n$ and $N$ are
varied.  In order to discern the dependence of the error on $N$ and
$n$, the values in Figure~\ref{table:MC} are normalized by the factor
$\sqrt{N/n}$. By replicating the Monte Carlo sampling we ensured that
the Monte Carlo sample size leads to the relative error of at most
6\%. The last column of the table reports the values of the above
bounds for the corresponding parameter values.

\begin{figure}
\begin{center}
\begin{tiny}
\begin{tabular}{|c|cccccc|c|}
\hline
\multicolumn{8}{|c|}{{\bf Design 1:} $x_A(q)=x^{(1:n)}(q)$ and $x_B(q)=q$.}\\
\hline
$n=$&\multicolumn{6}{|c|}{$N=$}&{\tiny Theorem 4.5}\\
        \cline{2-7}
         &$2$&$10^1$&$10^2$&$10^3$&$10^4$&$10^5$&{\tiny upper bound}\\
        \cline{2-8}
$2^2$	&	1.2534	&	0.3874	&	0.3625	&	0.3573	&	0.3535	&	0.3481	&	5.4221	\\
$2^3$	&	0.3915	&	0.4336	&	0.4829	&	0.4994	&	0.5041	&	0.4918	&	4.6957	\\
$2^4$	&	0.3025	&	0.3634	&	0.3982	&	0.4045	&	0.4088	&	0.4187	&	4.6166	\\
$2^5$	&	0.1831	&	0.2690	&	0.2825	&	0.2821	&	0.2798	&	0.2901	&	4.5622	\\
$2^6$	&	0.1437	&	0.2369	&	0.1929	&	0.1917	&	0.1898	&	0.1890	&	4.2406	\\
$2^7$	&	0.1316	&	0.1902	&	0.1590	&	0.1432	&	0.1407	&	0.1441	&	3.7786	\\
$2^8$	&	0.1276	&	0.1462	&	0.1490	&	0.1248	&	0.1226	&	0.1198	&	3.2644	\\
$2^9$	&	0.1247	&	0.1267	&	0.1488	&	0.1207	&	0.1100	&	0.1150	&	2.7543	\\
$2^{10}$	&	0.1215	&	0.1176	&	0.1616	&	0.1165	&	0.1114	&	0.1124	&	0.5702	\\
\hline\hline
\multicolumn{8}{|c|}{{\bf Design 2:} $x_A(q)=q$ and $x_B(q)=x^{(1:n)}(q)$.}\\
\hline
$n=$&\multicolumn{6}{|c|}{$N=$}&{\tiny Theorem 4.5}\\
        \cline{2-7}
         &$2$&$10^1$&$10^2$&$10^3$&$10^4$&$10^5$&{\tiny upper bound}\\
        \cline{2-8}
$2^2$	&	0.1416	&	0.0716	&	0.0586	&	0.0589	&	0.0578	&	0.0590	&	9.2103	\\
$2^3$	&	0.1281	&	0.0789	&	0.0637	&	0.0643	&	0.0670	&	0.0624	&	9.2103	\\
$2^4$	&	0.0790	&	0.0821	&	0.0596	&	0.0584	&	0.0591	&	0.0583	&	9.2103	\\
$2^5$	&	0.0427	&	0.0687	&	0.0488	&	0.0472	&	0.0452	&	0.0490	&	9.2103	\\
$2^6$	&	0.0224	&	0.0483	&	0.0369	&	0.0337	&	0.0337	&	0.0355	&	9.2103	\\
$2^7$	&	0.0116	&	0.0259	&	0.0290	&	0.0237	&	0.0232	&	0.0236	&	9.2103	\\
$2^8$	&	0.0059	&	0.0132	&	0.0230	&	0.0162	&	0.0157	&	0.0156	&	9.2103	\\
$2^9$	&	0.0030	&	0.0067	&	0.0186	&	0.0115	&	0.0107	&	0.0104	&	9.2103	\\
$2^{10}$	&	0.0015	&	0.0034	&	0.0107	&	0.0094	&	0.0084	&	0.0081	&	9.2103	\\
\hline\hline
\multicolumn{8}{|c|}{{\bf Design 3:} $x_A(q)=x^{(n-1:n)}(q)$ and $x_B(q)=x^{(1:n)}(q)$.}\\
\hline
$n=$&\multicolumn{6}{|c|}{$N=$}&{\tiny Theorem 4.5}\\
        \cline{2-7}
         &$2$&$10^1$&$10^2$&$10^3$&$10^4$&$10^5$&{\tiny upper bound}\\
        \cline{2-8}
$2^2$	&	0.1522	&	0.1978	&	0.2127	&	0.2104	&	0.2113	&	0.2101	&	9.2103	\\
$2^3$	&	0.1177	&	0.1200	&	0.1061	&	0.1056	&	0.1064	&	0.0998	&	9.2103	\\
$2^4$	&	0.0759	&	0.1116	&	0.1067	&	0.1046	&	0.1019	&	0.1029	&	9.2103	\\
$2^5$	&	0.0424	&	0.0866	&	0.0817	&	0.0798	&	0.0809	&	0.0800	&	9.2103	\\
$2^6$	&	0.0224	&	0.0499	&	0.0604	&	0.0602	&	0.0604	&	0.0581	&	9.2103	\\
$2^7$	&	0.0116	&	0.0259	&	0.0474	&	0.0443	&	0.0434	&	0.0447	&	9.2103	\\
$2^8$	&	0.0059	&	0.0132	&	0.0356	&	0.0323	&	0.0324	&	0.0322	&	9.2103	\\
$2^9$	&	0.0030	&	0.0067	&	0.0209	&	0.0226	&	0.0226	&	0.0235	&	9.2103	\\
$2^{10}$	&	0.0015	&	0.0034	&	0.0107	&	0.0171	&	0.0168	&	0.0161	&	9.2103\\
\hline\hline
\end{tabular}
\caption[!ht]{Mean absolute deviation for $\widehat{P}_{B}$ across Monte-Carlo simulations
normalized by $\sqrt{N}/n$.}
\label{table:MC}
\end{tiny}
\end{center}
\end{figure}

We make the following observations:
\begin{itemize}
\item {\bf Dependence on $N$:} Per our theoretical bound and
  normalization, we expect the values reported in the table to stay
  constant across different numbers of samples $N$. The table
  validates this: for all three designs, values along the rows in the
  table do not vary significantly.
\item {\bf Dependence on $n$:} Per our theoretical bound and
  normalization, we expect the values reported in the table to
  increase slowly with $n$ (logarithmically for Design 1, and
  sublinearly for Designs 2 and 3). We observe instead that values
  decrease with $n$ (along columns in the table). This indicates that
  the dependence of our theoretical upper bound on $n$ is loose and the
  true dependence of error on $n$ is much smaller.
\item  The table also directly contrasts the empirical performance
of our estimator with the theoretical upper bound given in Theorem 4.5.
The bound exceeds the measured mean absolute error, suggesting that 
it can be tightened further.
\end{itemize}

\begin{figure}
\begin{center}
\includegraphics[height=2.5in]{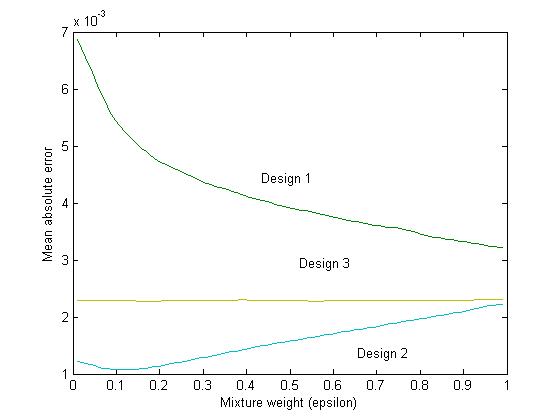}
\caption[!ht]{Dependence of the relative median absolute
error from the mixture weight $\epsilon$.}
\label{figure:eps}
\end{center}
\end{figure}

To illustrate the dependence of the estimation error on the choice of
the mixture weight $\epsilon$ for the three considered designs. We fix
the number of agents $n=32$ and the sample size $N=1000$ and 
vary mixture probability $\epsilon$ between 0 an 1, exclusive.  In
Figure~\ref{figure:eps} we demonstrate the dependence of the median
absolute error computed as the ratio of the median absolute error to
estimated revenue.  The figure demonstrates that for Designs~1 and~3,
error decreases or remains constant as $\epsilon$ increases. This is
consistent with our theoretical bound. On the other hand, for
Design~2, after a small initial decrease, the error surprisingly
increases with $\epsilon$. This increase is not directly captured by
our theoretical bound.\footnote{Our theoretical bound nevertheless
  applies at all values of $\epsilon$; the seeming discrepancy arises
  because relative to the actual error, our bound is looser for small
  $\epsilon$ than for large $\epsilon$.}  Intuitively, the auction A
in Design 2 is much better for inferring the revenue of B than B
itself, because the bid distribution it generates has high density on
its entire range.


\section{Applications to A/B testing}
\label{s:ab-testing}

We now discuss applications of the inference approach we developed in Section~\ref{s:param-inf}.

\subsection{Estimating revenues of novel mechanisms}
\label{s:ab-revenue}

Recall the setting of an ideal A/B test: an auction house running
auction $A$ and wanting to determine the revenue of a novel mechanism
$B$ mixes in the auction $B$ with some probability $\eps>0$. Suppose
that in doing so, the auction house obtains $\eps N$ bids in response
to the auction $B$ out of a total of $N$ bids, the revenue of $B$ can
be estimated within an error bound of
\begin{align}
\label{error-ideal}
\Theta\left(\frac 1{\sqrt{\eps}}\right)\frac{\sup_{\quant}\{x_B'(\quant)\}}{\sqrt{N}}.
\end{align}
In practice, however, instead of obtaining bids in response to $B$
alone, the seller obtains bids in response to the aggregate mechanism
$C=(1-\eps) A +\eps B$. We can then use \eqref{AB-test-estimator-sum}
to estimate the revenue of $B$.
As a consequence of Theorem~\ref{thm:inference-y}, and noting that for
positions auctions with $n$ positions, $x_C'(\quant)\le n$ and
$x_B'(\quant)/x_C'(\quant)\le 1/\eps$ for all quantiles $\quant$, we
obtain the following error bound.
\begin{corollary}
\label{cor1} 
The revenue of a rank based mechanism $B$ can be estimated from $N$
bids of a rank-based mechanism $C=(1-\eps)A+\eps B$ with absolute error
bounded by
\begin{align}
\label{error-true}
O(1) \sqrt{n\log n} \log (n/\eps)\,\,\frac{\sup_{\quant}\{x_B'(\quant)\}}{\sqrt{N}}.
\end{align}
\end{corollary}

Relative to the ideal situation described above, our error bound has a
better dependence on $\eps$ and a worse dependence on $n$. Note that
when $\eps$ is very small, our error bound \eqref{error-true} may be
smaller than the ideal bound in \eqref{error-ideal}. This is not
surprising: the ideal bound ignores information that we can learn
about the revenue of $B$ from the $(1-\eps)N$ bids obtained when $B$
is not run.


When $B$ is a multi-unit auction, we obtain a better error bound
using Theorem~\ref{th: all pay} which is closer to the ideal bound in \eqref{error-ideal}.
\begin{corollary}
\label{cor2}
  The revenue of the highest-$k$-bids-win mechanism $B$ can be
  estimated from $N$ bids of a rank-based mechanism $C=(1-\eps)A+\eps
  B$ with absolute error bounded by
\begin{align}
\label{error-true-k}
40 \log (n/\eps)\,\,\frac{\sup_{\quant}\{x_B'(\quant)\}}{\sqrt{N}}.
\end{align}
\end{corollary}

\subsection{Universal B test}
\label{s:universal}

We now consider the problem of estimating all of the multi-unit
revenues $\murevk$ simultaneously from the bids of a single
auction. What properties should the auction C have in order to enable
this?  The following definition formalizes this notion.

\begin{definition} An auction B is a {\em universal B test} if 
for any rank-based auction A, any $\eps>0$, and auction C defined by
$x_C = (1-\eps)x_A +\eps x_B$; all of the multi-unit revenues
$\murevk[1],\ldots,\murevk[n]$ can be estimated from bids of C with
the dependence of the absolute error on $\eps$ bounded by
$O(\log(1/\eps))$.
\end{definition}



In \citet{CHN-14} we showed that it suffices to mix the $k$-unit auction for
every $k$ into C with some small probability. In other words, the
position auction B with position weights $w_k=1-k/n$ is a universal B test.
We will now prove that in fact we can get similar results by mixing in
just a few of the multi-unit auctions. In particular, in order to
estimate $\murevk$ accurately, it suffices to mix in a multi-unit
auction with no more than $k$ units, and another one with no less than
$k$ units. This gives us a more efficient universal B test for
simultaneously inferring all of the multi-unit revenues (see
Corollary~\ref{cor:universal}).

\begin{lemma}
\label{lem:univ}
  The revenue of the highest-$k$-bids-win mechanism B can be
  estimated from $N$ bids of a rank-based all-pay auction C $=$ $(1-2\eps)$A $+\eps
  $B$_1+\eps $B$_2$ where A is an arbitrary rank-based auction, and
  B$_1$ and B$_2$ are the highest-$k_1$-bids-win and
  highest-$k_2$-bids-win auctions respectively, with $k_1\le k\le
  k_2$. The absolute error of the estimate is bounded by
\begin{align*}
40 (n+\log (1/\eps))\,\,\frac{\sup_{\quant}\{x_k'(\quant)\}}{\sqrt{N}}.
\end{align*}
\end{lemma}
\begin{proof}
We begin by noting that for any $j$ and $k$ with $k\le j$,
$$
\frac{\alloc'_k(\quant)}{\alloc'_j(\quant)} = \frac{{j-1\choose
    k-1}}{{{n-1-k}\choose {n-1-j}}} \left(\frac{q}{1-q}\right)^{j-k}.
$$
When $k\le j$ and $q\ge 1/2$, this ratio is less than $2^n$. Likewise,
we can show that when $k\ge j$ and $q\le 1/2$, the ratio
$\frac{\alloc'_k(\quant)}{\alloc'_j(\quant)} $ is less than
$2^n$. Therefore, for any $q$, and C $=(1-2\eps)$A $+\eps$B$_1+\eps$B$_2$
where B$_1$ and B$_2$ are the highest-$k_1$-bids-win and
highest-$k_2$-bids-win auctions respectively, with $k_1\le k\le k_2$,
we have
$$
\frac{\alloc'_k(\quant)}{\alloc'_C(\quant)} \le \frac{1}{\epsilon} 2^n.
$$
Next we note that $\sup_q \alloc'_C(q)\le n$, and therefore, $\sup_{\quant:
    \kalloc'(\quant)\ge 1}\frac{\alloc_C'(\quant)}{\kalloc'(\quant)}\le n$.
Putting these quantities together with Theorem~\ref{th: all pay}, we
get that the absolute error in estimating $\murevk$ from bids drawn from C is at
most 
\begin{align*}
\Err{\murevk}& \le \frac{40}{\sqrt{\samples}}\,\,
\sup_{\quant}\{\kalloc'(\quant)\}\,\, \left(n+\log 1/\eps\right). \qedhere
\end{align*}
\end{proof}

\begin{corollary}
\label{cor:universal}
Let B be the position auction with position weights $w_1=1$, $w_k=1/2$ for
$1< k<n-1$, and $w_n=0$. Then
$N$ bids from a mechanism C with $x_C = (1-\eps)x_A + \eps x_B$, where
A is an arbitrary rank-based auction, can be used to simultaneously
estimate {\em all} of the multi-unit revenues, and consequently all
position auction revenues with absolute error bounded by
\begin{align*}
\frac{40 n(n+\log (1/\eps))}{\sqrt{N}}.
\end{align*}
  \end{corollary}

\subsection{Comparing revenues}
\label{s:direct-comparison}

We have considered the case where the empirical task was to recover
the revenues for one mechanism ($y$) using the sample of bids
responding to another mechanism ($x$). In many practical situations
the empirical task is simply the verification of whether the revenue
from a given mechanism is higher than the revenue from another
mechanism. Or, equivalently, the task could be to verify whether one
mechanism provides revenue which is a certain percentage above that of
another mechanism. We now demonstrate that this is a much easier
empirical task in terms of accuracy than the task of inferring the
revenue.
\par
Suppose that we want to compare the revenues of mechanisms $B_1$ and
$B_2$ by mixing them in to an incumbent mechanism $A$, and running the
composite mechanism $C = \eps B_1 + \eps B_2 + (1-2\eps) A$.
Specifically, we would like to determine whether $\REV{B_1}>\alpha
\REV{B_2}$ for some $\alpha>0$. Consider a binary classifier
$\widehat{\gamma}$ which is equal to $1$ when $\REV{B_1}>\alpha
\REV{B_2}$ and $0$ otherwise. Let $\gamma={\bf
  1}\{P_{B_1}-\alpha\,P_{B_2}>0\}$ be the corresponding ``ideal"
classifier for the case where the distribution of bids from mechanism
$C$ is known precisely. To evaluate the accuracy of the classifier, we
need to evaluate the probability $\Pr(\widehat{\gamma}=1|\gamma=0)$,
and likewise, $\Pr(\widehat{\gamma}=0|\gamma=1)$. The classifier will
give the wrong output if the sampling noise in estimating
$\hat{P}_{B_1}-\alpha\,\hat{P}_{B_2}$ is greater than
$|P_{B_1}-\alpha\,P_{B_2}|$.

Our main result of this section says that keeping $\alpha$, the
difference $|P_{B_1}-\alpha\,P_{B_2}|$, and the number of positions
$n$ constant, the probability of incorrect output decreases
exponentially with the number of bids $N$. See
Appendix~\ref{s:app-proofs} for a proof.

\begin{theorem}
\label{theorem:sign-AB}
Suppose that $N$ bids from a mechanism $C=\eps B_1+\eps B_2+(1-2\eps)A$ for
arbitrary rank-based mechanisms $B_1$, $B_2$, and $A$, are used to
estimate the classifier $\gamma={\bf 1}\{P_{B_1}-\alpha\,P_{B_2}>0\}$ that
establishes whether the revenue of mechanism $B_1$ exceeds $\alpha$
times the revenue of mechanism $B_2$.  Then the error rate of the binary
classifier is bounded from above by
$$\exp\left(-O\left(\frac{Na^2}{\alpha^2n^3\log(n/\eps)}\right)\right),$$
where $a=|P_{B_1}-\alpha\,P_{B_2}|$.  In other words, once the number of
samples is polynomially large in $n$, the error rate decreases
exponentially with the number of samples.
\end{theorem}

We obtain a similar error bound when our goal is to estimate which of
$r$ different novel mechanisms obtains the most revenue, for any
$r>1$:

\begin{corollary}
  Suppose that our goal is to determine which of $r$ position
  auctions, $B_1, B_2, \cdots, B_r$, obtains the most revenue while
  running incumbent mechanism $A$, by running each of the novel
  mechanisms with probability $\eps/r$. Then the error probability of
  the corresponding classifier constructed using $N$ bids from
  composite mechanism $C = \sum_{i=1}^{r} \eps/r B_i + (1-\eps)A$ is
  bounded from above by
$$r\exp\left(-O\left(\frac{Na^2}{n^3\log(rn/\eps)}\right)\right),$$
where $a$ is the absolute difference between the revenue obtained by
the best two of the $r$ mechanisms.
\end{corollary}




\bibliographystyle{acmsmall}
\bibliography{uniq,agt,references}

\begin{thebibliography}{}

\bibitem[\protect\citeauthoryear{Athey and Haile}{Athey and
  Haile}{2007}]{athey:2007}
{\sc Athey, S.} {\sc and} {\sc Haile, P.~A.} 2007.
\newblock Nonparametric approaches to auctions.
\newblock {\em Handbook of econometrics\/}~{\em 6}, 3847--3965.

\bibitem[\protect\citeauthoryear{Baliga and Vohra}{Baliga and
  Vohra}{2003}]{BV-03}
{\sc Baliga, S.} {\sc and} {\sc Vohra, R.} 2003.
\newblock Market research and market design.
\newblock {\em Advances in Theoretical Economics\/}~{\em 3,\/}~1.

\bibitem[\protect\citeauthoryear{Blum and Hartline}{Blum and
  Hartline}{2005}]{BH-05}
{\sc Blum, A.} {\sc and} {\sc Hartline, J.~D.} 2005.
\newblock Near-optimal online auctions.
\newblock In {\em Proceedings of the sixteenth annual ACM-SIAM symposium on
  Discrete algorithms}. Society for Industrial and Applied Mathematics,
  1156--1163.

\bibitem[\protect\citeauthoryear{Brown and Morgan}{Brown and
  Morgan}{2009}]{BM-09}
{\sc Brown, J.} {\sc and} {\sc Morgan, J.} 2009.
\newblock How much is a dollar worth? tipping versus equilibrium coexistence on
  competing online auction sites.
\newblock {\em Journal of Political Economy\/}~{\em 117,\/}~4, 668--700.

\bibitem[\protect\citeauthoryear{Cesa-Bianchi, Gentile, and
  Mansour}{Cesa-Bianchi et~al\mbox{.}}{2015}]{CGM-15}
{\sc Cesa-Bianchi, N.}, {\sc Gentile, C.}, {\sc and} {\sc Mansour, Y.} 2015.
\newblock Regret minimization for reserve prices in second-price auctions.
\newblock {\em IEEE Transactions on Information Theory\/}~{\em 61,\/}~1, 549.

\bibitem[\protect\citeauthoryear{Chawla, Hartline, and Nekipelov}{Chawla
  et~al\mbox{.}}{2014}]{CHN-14}
{\sc Chawla, S.}, {\sc Hartline, J.}, {\sc and} {\sc Nekipelov, D.} 2014.
\newblock Mechanism design for data science.
\newblock In {\em Proceedings of the fifteenth ACM conference on Economics and
  computation}. ACM, 711--712.

\bibitem[\protect\citeauthoryear{Chawla and Hartline}{Chawla and
  Hartline}{2013}]{CH-13}
{\sc Chawla, S.} {\sc and} {\sc Hartline, J.~D.} 2013.
\newblock Auctions with unique equilibria.
\newblock In {\em Proceedings of the Fourteenth ACM Conference on Electronic
  Commerce}. EC '13. ACM, New York, NY, USA, 181--196.

\bibitem[\protect\citeauthoryear{Cheng and Parzen}{Cheng and
  Parzen}{1997}]{cheng:97}
{\sc Cheng, C.} {\sc and} {\sc Parzen, E.} 1997.
\newblock Unified estimators of smooth quantile and quantile density functions.
\newblock {\em Journal of statistical planning and inference\/}~{\em 59,\/}~2,
  291--307.

\bibitem[\protect\citeauthoryear{Cole and Roughgarden}{Cole and
  Roughgarden}{2014}]{CR-14}
{\sc Cole, R.} {\sc and} {\sc Roughgarden, T.} 2014.
\newblock The sample complexity of revenue maximization.
\newblock In {\em Proceedings of the 46th Annual ACM Symposium on Theory of
  Computing}. ACM, 243--252.

\bibitem[\protect\citeauthoryear{Cs{\"o}rg{\"o}}{Cs{\"o}rg{\"o}}{1983}]{csorgo:83}
{\sc Cs{\"o}rg{\"o}, M.} 1983.
\newblock {\em Quantile processes with statistical applications}.
\newblock SIAM.

\bibitem[\protect\citeauthoryear{Csorgo and Revesz}{Csorgo and
  Revesz}{1978}]{csorgo:78}
{\sc Csorgo, M.} {\sc and} {\sc Revesz, P.} 1978.
\newblock Strong approximations of the quantile process.
\newblock {\em The Annals of Statistics\/}, 882--894.

\bibitem[\protect\citeauthoryear{Edelman, Ostrovsky, and Schwarz}{Edelman
  et~al\mbox{.}}{2007}]{EOS-07}
{\sc Edelman, B.}, {\sc Ostrovsky, M.}, {\sc and} {\sc Schwarz, M.} 2007.
\newblock Internet advertising and the generalized second-price auction:
  Selling billions of dollars worth of keywords.
\newblock {\em The American Economic Review\/}~{\em 97,\/}~1, 242--259.

\bibitem[\protect\citeauthoryear{Fu, Haghpanah, Hartline, and Kleinberg}{Fu
  et~al\mbox{.}}{2014}]{FHHK-14}
{\sc Fu, H.}, {\sc Haghpanah, N.}, {\sc Hartline, J.}, {\sc and} {\sc
  Kleinberg, R.} 2014.
\newblock Optimal auctions for correlated buyers with sampling.
\newblock In {\em Proceedings of the fifteenth ACM conference on Economics and
  computation}. ACM, 23--36.

\bibitem[\protect\citeauthoryear{Goldberg, Hartline, Karlin, Saks, and
  Wright}{Goldberg et~al\mbox{.}}{2006}]{GHKSW-06}
{\sc Goldberg, A.~V.}, {\sc Hartline, J.~D.}, {\sc Karlin, A.~R.}, {\sc Saks,
  M.}, {\sc and} {\sc Wright, A.} 2006.
\newblock Competitive auctions.
\newblock {\em Games and Economic Behavior\/}~{\em 55,\/}~2, 242--269.

\bibitem[\protect\citeauthoryear{Gomes and Sweeney}{Gomes and
  Sweeney}{2009}]{GS-09}
{\sc Gomes, R.} {\sc and} {\sc Sweeney, K.} 2009.
\newblock Bayes-nash equilibria of the generalized second price auction.
\newblock In {\em Proceedings of the 10th ACM conference on Electronic
  commerce}. ACM, 107--108.

\bibitem[\protect\citeauthoryear{Guerre, Perrigne, and Vuong}{Guerre
  et~al\mbox{.}}{2000}]{guerre}
{\sc Guerre, E.}, {\sc Perrigne, I.}, {\sc and} {\sc Vuong, Q.} 2000.
\newblock Optimal nonparametric estimation of first-price auctions.
\newblock {\em Econometrica\/}~{\em 68,\/}~3, 525--574.

\bibitem[\protect\citeauthoryear{Hartline}{Hartline}{2013}]{har-13}
{\sc Hartline, J.~D.} 2013.
\newblock Bayesian mechanism design.
\newblock {\em Foundations and Trends in Theoretical Computer Science\/}~{\em
  8,\/}~3, 143--263.

\bibitem[\protect\citeauthoryear{Kleinberg and Leighton}{Kleinberg and
  Leighton}{2003}]{KL-03}
{\sc Kleinberg, R.} {\sc and} {\sc Leighton, T.} 2003.
\newblock The value of knowing a demand curve: Bounds on regret for online
  posted-price auctions.
\newblock In {\em Foundations of Computer Science, 2003. Proceedings. 44th
  Annual IEEE Symposium on}. IEEE, 594--605.

\bibitem[\protect\citeauthoryear{Myerson}{Myerson}{1981}]{mye-81}
{\sc Myerson, R.} 1981.
\newblock Optimal auction design.
\newblock {\em Mathematics of Operations Research\/}~{\em 6}, 58--73.

\bibitem[\protect\citeauthoryear{Ostrovsky and Schwarz}{Ostrovsky and
  Schwarz}{2011}]{OS-11}
{\sc Ostrovsky, M.} {\sc and} {\sc Schwarz, M.} 2011.
\newblock Reserve prices in internet advertising auctions: A field experiment.
\newblock In {\em Proceedings of the 12th ACM Conference on Electronic
  Commerce}. EC '11. ACM, New York, NY, USA, 59--60.

\bibitem[\protect\citeauthoryear{Paarsch and Hong}{Paarsch and
  Hong}{2006}]{paarsch}
{\sc Paarsch, H.~J.} {\sc and} {\sc Hong, H.} 2006.
\newblock {\em An introduction to the structural econometrics of auction data}.
  Vol.~1.
\newblock The MIT Press.

\bibitem[\protect\citeauthoryear{Reiley}{Reiley}{2006}]{Ril-06}
{\sc Reiley, D.~H.} 2006.
\newblock Field experiments on the effects of reserve prices in auctions: more
  magic on the internet.
\newblock {\em The RAND Journal of Economics\/}~{\em 37,\/}~1, 195--211.

\bibitem[\protect\citeauthoryear{Segal}{Segal}{2003}]{seg-03}
{\sc Segal, I.} 2003.
\newblock Optimal pricing mechanisms with unknown demand.
\newblock {\em The American economic review\/}~{\em 93,\/}~3, 509--529.

\bibitem[\protect\citeauthoryear{Varian}{Varian}{2007}]{var-06}
{\sc Varian, H.} 2007.
\newblock Position auctions.
\newblock {\em International Journal of Industrial Organization\/}~{\em
  25,\/}~6, 1163--1178.

\end{thebibliography}


\appendix

\section{Inference for Auctions}
\label{app:inference}

The distribution of values, which is unobserved, can be inferred from
the distribution of bids, which is observed.  Once the value
distribution is inferred, other properties of the value distribution
such as its corresponding revenue curve, which is fundamental for
optimizing revenue, can be obtained.  In this section we describe the
basic premise of the inference assuming that the distribution of bids
known exactly.

The key idea behind the inference of the value distribution from the
bid distribution is that the strategy which maps values to bids is a
best response, by equation~\eqref{eq:br}, to the distribution of bids.
As the distribution of bids is observed, and given suitable continuity
assumptions, this best response function can be inverted.

The value distribution can be equivalently specified by distribution
function $\dist(\cdot)$ or value function $\val(\cdot)$; the bid
distribution can similarly be specified by the bid function
$\bid(\cdot)$.  For rank-based auctions (as considered by this paper)
the allocation rule $\alloc(\cdot)$ in quantile space is known
precisely (i.e. it does not depend on the value function
$\val(\cdot)$).  Assume these functions are monotone, continuously
differentiable, and invertible.

\paragraph{Inference for first-price auctions}
Consider a first-price rank-based auction with a symmetric bid
function $\bid(\quant)$ and allocation rule $\alloc(\quant)$ in BNE.
To invert the bid function we solve for the bid that the agent with
any quantile would make.  Continuity of this bid function implies that
its inverse is well defined.  Applying this inverse to the bid
distribution gives the value distribution.

The utility of an agent with quantile $\quant$ as a function of his bid $z$
is
\begin{align*}
\yestag\label{eq:fp-util}
\util(\quant,z) &= (\val(\quant) - z) \, \alloc(\bid^{-1}(z)).\\ 
\intertext{Differentiating with respect to $z$ we get:} 
\tfrac{d}{dz}\util(\quant,z) &= -\alloc(\bid^{-1}(z)) +
\big(\val(\quant) - z\big) \, \alloc'(\bid^{-1}(z))\,
\tfrac{d}{dz}\bid^{-1}(z).\\ 
\intertext{Here $\alloc'$ is the
  derivative of $\alloc$ with respect to the quantile $q$. Because
  $\bid(\cdot)$ is in BNE, the derivative $\tfrac{d}{dz}\util(\quant,z)$ is $0$ at
  $z=\bid(\quant)$. Rarranging, we obtain:}
  \val(\quant) &= \bid(\quant) + \tfrac{\alloc(\quant)\,\bid'(\quant)}{\alloc'(\quant)}
\end{align*}

\paragraph{Inference for all-pay auctions}
We repeat the calculation above for rank-based all-pay auctions; the starting
equation \eqref{eq:fp-util} is replaced with the analogous equation for all-pay auctions:
\begin{align*}
\yestag\label{eq:ap-util}
\util(\quant,z) &= \val(\quant)\,\alloc(\bid^{-1}(z)) - z.
\intertext{Differentiating with respect to $z$ we obtain:}
\tfrac{d}{dz}\util(\quant,z) &= \val(\quant)\,\alloc'(\bid^{-1}(z)) \,\frac{d}{dz}\bid^{-1}(z) - 1,\\
\intertext{Again the first-order condition of BNE implies that this expression is zero at $z = \bid(\quant)$; therefore,}
  \val(\quant) & = \tfrac{\bid'(\quant)}{\alloc'(\quant)}.
\end{align*}

\paragraph{Known and observed quantities}
Recall that the functions $\alloc(\quant)$ and $\alloc'(\quant)$ are
known precisely: these are determined by the rank-based auction
definition.  The functions $\bid(\quant)$ and $\bid'(\quant)$ are
observed.  The calculations above hold in the limit as the number of
samples from the bid distribution goes to infinity, at which point these
obserations are precise.

\section{Inference for social welfare}
\label{s:welfare}

We now consider the problem of estimating the social welfare of a
rank-based auction using bids from another rank-based all-pay
auction. Consider a rank-based auction with induced position weights
$\wals$. By definition, the expected {\em per-agent} social welfare
obtained by this auction is as below, where $\evalk$ is the expected
value of the $k$th highest value agent, or the $k$th order statistic
of the value distribution.

$$
\sw = \frac 1n \sum_{k=1}^{n} \walk \evalk
$$
 
We note that the value order statistics, $\evalk$, are closely
related to the expected revenues of the multi-unit auctions. The
$k$-unit second-price auction serves the top $k$ agents with
probability $1$, and charges each agent the $k+1$th highest value. Its
expected revenue is therefore $nP_k = k\evalk[k+1]$. We
therefore obtain:

$$
\sw = \walk[1]\frac{\evalk[1]}{n} + \sum_{k=1}^{n-1} \walk[k+1] \, \frac {P_{k}}{k}
$$
 
The methodology developed in the previous sections can be used to
estimate the $P_k$'s in the above expression. The first order
statistic of the values, $\evalk[1]$, cannot be directly estimated in
this manner. Notate the expected value of an agent as
$$
\expval = \expect[\quant]{\val(\quant)} = \frac 1 n \sum_{k=1}^n \evalk
$$
where $\expval$ is the expected value of any one agent. Therefore, we can calculate the social welfare of the position auction with weights $\wals$ as
\begin{eqnarray}
\sw = \walk[1] \expval - \sum_{k=2}^{n} (\walk[1]-\walk) \frac
{\evalk}{n} = \walk[1]\expval - \sum_{k=1}^{n-1} (\walk[1]-\walk[k+1]) \, \frac {P_k}{k}
\label{eq:sw}
\end{eqnarray}

We now argue that $\expval$ can be estimated at a good rate from the
bids of another rank-based all-pay auction. Let $\alloc$ denote the allocation rule of the auction that we run,
and $\bid$ denote the bid distribution in BNE of this auction. Then we
note that
$$
\expval = \expect[\quant]{\val(\quant)} =
\expect[\quant]{\frac{\bid'(\quant)}{\alloc'(\quant)}} = \expect[\quant]{-\Zbar'(\quant)\bid(\quant)}
$$
where $\Zbar(q) = 1/\alloc'(q)$. We might now try to directly apply
Theorem~\ref{th: all pay} to bound
the error in our estimate of $\expval$. This does not work, as
Lemma~\ref{lem:Z-bound-1} fails to hold for $\Zbar$. Instead, we can
follow the argument in the proof of Theorem~\ref{thm:inference-y} and
obtain the following lemma:

\begin{lemma}
\label{lem:inference-expval}
The expected absolute error in estimating the expected value $\expval$
using $N$ samples from the bid distribution for an all-pay position
auction with allocation rule $x$ is bounded as below; Here $n$ is the
number of positions in the position auction.
\begin{align*}
\Err{\expval}
& \le \tfrac{40}{\sqrt{N}} \sqrt{n\log n} \,\,\log \max\left\{\sup\nolimits_{\quant}\alloc'(\quant),
\sup\nolimits_{\quant}\tfrac{1}{\alloc'(\quant)} \right\} \\
& \quad + \tfrac{O(1)}{N} \sup\nolimits_{\quant}\{\alloc'(\quant)\}\,\,\sup\nolimits_{\quant}
\left\{ \tfrac{1}{\alloc'(\quant)} \right\}
\end{align*}
\end{lemma}

As an example application of Lemma~\ref{lem:inference-expval}, we
adapt Corollary~\ref{cor:universal} to bound the error of from
estimating the social welfare of any position auction that is mixed
with the universal B test mechanism of Section~\ref{s:universal}.
Other revenue estimation results can be similarly adapted to estimate
social welfare.   
\begin{theorem}
\label{thm:sw}
Let B be the position auction with position weights $w_1=1$, $w_k=1/2$
for $1< k<n-1$, and $w_n=0$. Then $N$ bids from a mechanism C with
$x_C = (1-\eps)x_A + \eps x_B$, where A is an arbitrary rank-based
auction, can be used to simultaneously estimate the per-agent social
welfare of any rank-based mechanism with absolute error bounded by
\begin{align*}
\frac{40 n\log n(n+\log (1/\eps))}{\sqrt{N}} + \frac{40\sqrt{n\log
      n}\log (n/\eps)}{\sqrt{N}} + O\left(\frac{n}{\epsilon
    N}\right) = O\left( \frac{n\log n(n+\log (1/\eps))}{\sqrt{N}} \right) + O\left(\frac{n}{\epsilon
    N}\right).
\end{align*}
\end{theorem}
The theorem follows by combining Lemma~\ref{lem:inference-expval}
with equation~\eqref{eq:sw} and Corollary~\ref{cor:universal}.  In the
statement of the theorem the error term corresponding to the $P_k$'s
has an extra factor of $\log n$ (relative to the statement of
Corollary~\ref{cor:universal}) because the total weight of the
multipliers for the terms in equation~\eqref{eq:sw} can be as large
as $\log n$.

\section{Inference methodology and error bounds for first-price auctions}
\label{s:fp-inf}

While most of our analysis so far has focused on the case of all-pay
auctions, our methodology and results extend as-is to first-price
auctions as well. Here we sketch the differences between the two cases.

Recall that in a first-price auction, we can obtain the value
distribution from the bid distribution as follows: $\val(\quant) =
\bid(\quant) +
\alloc(\quant)\bid'(\quant)/\alloc'(\quant)$. Substituting this into
the expression for $P_y$ we get:
$$
P_y= \expect[\quant]{(1-\quant)y'(\quant) \bid(\quant) +
  \frac{(1-\quant)y'(\quant)
    \alloc(\quant)\bid'(\quant)}{\alloc'(\quant)}} =
\expect[\quant]{(1-\quant)y'(\quant) \bid(\quant) + Z_y(\quant) \alloc(\quant)\bid'(\quant)}
$$
where, as before, $Z_y(\quant) =
\frac{(1-\quant)y'(\quant)}{\alloc'(\quant)}$.

Integrating the second expression by parts, we get
\begin{eqnarray*}
\int_0^1 Z_y(\quant) \alloc(\quant)\bid'(\quant)\,d\quant & = &
Z_y(\quant) \alloc(\quant)\bid(\quant)|_0^1 - \int_0^1 (Z'_y(\quant) \alloc(\quant)+Z_y(\quant) \alloc'(\quant))\bid(\quant) \,d\quant\\
& = & - \int_0^1 Z'_y(\quant) \alloc(\quant)\bid(\quant)\,d\quant -  \int_0^1 (1-\quant) y'(\quant)\bid(\quant)\,d\quant 
\end{eqnarray*}

When we put this back in the expression for $P_y$ two of the terms
cancel, and we get the following lemma.
\begin{lemma}
  The per-agent revenue of an auction with allocation rule $y$ can be written as
  a linear combination of the bids in a first-pay auction: 
$$
P_y = \expect[\quant]{-\alloc(\quant)Z'_y(\quant)\bid(\quant)}
$$
where $Z_y(\quant)
=(1-\quant)\frac{y'(\quant)}{\alloc'(\quant)}$ and
$\alloc(\quant)$ are known precisely.
\end{lemma}

As in the case of the all-pay auction format, we can write the error
in $P_y$ as:
\begin{eqnarray*}
|\Phat_y-P_y| & = &
\expect[q]{\left|-\alloc(\quant)Z'_y(\quant)(\widehat{\bid}(\quant)-\bid(\quant))\right|}\\
& \le & E\left[\frac{\left(\log(1+Z_y(q))\right)^{\alpha}}{Z_y(q)}|Z'_y(q)|\right]
\sup\limits_q\left|\alloc(\quant)\frac{Z_y(q)}{\left(\log(1+Z_y(q))\right)^{\alpha}}(\hat{b}(q)-b(q))\right|
\end{eqnarray*}
With an appropriate choice of $\alpha$ we obtain the following theorem.
\begin{theorem}\label{th: first price}
  The expected absolute error in estimating the revenue of a position
  auction with allocation rule $y$ using $N$ samples from the bid
  distribution for a first-pay position auction with allocation rule
  $x$ is bounded as below; Here $n$ is the number of positions
  in the two position auctions.
\begin{align*}
\Err{P_y}
& \le \frac{40}{\sqrt{N}} \sqrt{n\log n} \,\,\sup_{\quant}\{y'(\quant)\}\,\,\log \max\left\{\sup_{\quant: y'(\quant)\ge 1}\frac{\alloc'(\quant)}{y'(\quant)},
\sup_{\quant}\frac{y'(\quant)}{\alloc'(\quant)} \right\} \\
& + \frac{O(1)}{N} \sup_{\quant}\{\alloc'(\quant)\}\,\,\sup_{\quant}
\left\{ \frac{y'(\quant)}{\alloc'(\quant)} \right\}.
\end{align*}
When $y$ is the highest-$k$-bids-win allocation rule, the error
improves to:
\begin{align*}
\Err{P_k}
& \le \frac{40}{\sqrt{N}} \,\,\sup_{\quant}\{\kalloc'(\quant)\}\,\,\log \max\left\{\sup_{\quant: \kalloc'(\quant)\ge 1}\frac{\alloc'(\quant)}{\kalloc'(\quant)},
\sup_{\quant}\frac{\kalloc'(\quant)}{\alloc'(\quant)} \right\}.
\end{align*}
\end{theorem}

Because the error bounds in Theorem~\ref{th: first price} are
identical to those in Theorems~\ref{th: all pay} and
\ref{thm:inference-y}, Lemma~\ref{lem:univ}, Corollaries~\ref{cor1},
\ref{cor2}, \ref{cor:universal},
and Theorems~\ref{theorem:sign-AB} and \ref{thm:sw} continue to hold when bids are drawn
from a first-price auction.

\section{Missing proofs}
\label{s:app-proofs}

\begin{fact}
\label{fact:max-slope}
The maximum slope of the allocation rule $\alloc$ of any $n$-agent
position auction is bounded by $n$: $\sup_q \alloc'(q)\le n$. The
maximum slope of the allocation rule $\kalloc$ for the $n$-agent
highest-$k$-bids-win auction is bounded by $$\sup_q \kalloc'(q)\in \left[\frac{1}{\sqrt{2\pi}},\frac{1}{\sqrt{\pi}}\right]
\frac{n-1}{\sqrt{\min\{k-1, n-k\}}} = \Theta\left(\frac n{\sqrt{\min\{k,n-k\}}}\right).$$
\end{fact}

\begin{proofof}{Theorem~\ref{th: all pay}}
We will make use of the following lemma from \citet{CHN-14}.
\begin{lemma}[\citet{CHN-14} ]
\label{lem:Z-bound-1}
Let $\kalloc$ denote the allocation function of the
$k$-highest-bids-win auction and $\alloc$ be any convex combination
over the allocation functions of the multi-unit auctions. Then the
function $Z_k(\quant)=(1-\quant)
\frac{\kalloc'(\quant)}{\alloc'(\quant)}$ achieves a single local
maximum for $\quant\in [0,1]$.
\end{lemma}

Recall that for $\alpha>0$ we can write\footnote{In this entire proof
  the logarithms are natural logarithms.}
$$
|\hat{P}_k-P_k|\leq E\left[\frac{\left(\log(1+Z_k(q))\right)^{\alpha}}{Z_k(q)}|Z'_k(q)|\right]
\sup\limits_q\left|\frac{Z_k(q)}{\left(\log(1+Z_k(q))\right)^{\alpha}}(\hat{b}(q)-b(q))\right|
$$
We start by considering the first term.
Lemma~\ref{lem:Z-bound-1} shows that $Z'_k(\cdot)$ changes sign
only once. Consider the region where the sign of $Z'_k(\cdot)$ is constant
and make the change of variable $t=Z_k(q)$. 
Recall that $Z_k^*=\sup_qZ_k(q)$ and we note that $\inf_qZ_k(q) \geq 0$.
Then we can evaluate the first term as
$$
E\left[\frac{\left(\log(1+Z_k(q))\right)^{\alpha}}{Z_k(q)}|Z'_k(q)|\right] \leq
2 \int^{Z_k^*}_0\frac{(\log\,(1+t))^{\alpha}}{t}\,dt
$$
Note that for any $t>0$, $\log(1+t)\le t$. Thus,
$$
 \int^{\delta}_0\frac{(\log\,(1+t))^{\alpha}}{t}\,dt< \frac{\delta^{\alpha}}{\alpha}
$$
Now split the integral into two pieces as
$$
\int^{Z_k^*}_0\frac{(\log\,(1+t))^{\alpha}}{t}\,dt=\int^{1}_0\frac{(\log\,(1+t))^{\alpha}}{t}\,dt+\int^{Z_k^*}_1\frac{(\log\,(1+t))^{\alpha}}{t}\,dt
$$
We just proved that the first piece is at most $1/\alpha$. Now we
upper bound the second piece and consider the integrand at $t \ge 1$. First, note that 
$$
(\log\,(1+t))^{\alpha}=\left(\log\,t+\log(1+\frac{1}{t})\right)^{\alpha}\le\left(\log\,t+\frac{1}{t}\right)^{\alpha}\le(\log\,t+1)^{\alpha}.
$$
Thus, the integral behaves as
$$
\int^{Z_k^*}_1\frac{(\log\,(1+t))^{\alpha}}{t}\,dt\le\int^{Z_k^*}_1\frac{(\log\,(t)+1)^{\alpha}}{t}\,dt=\frac{1}{1+\alpha}(\log\,Z_k^*+1)^{1+\alpha}.
$$
Thus, we just showed that 
$$
E\left[\frac{\left(\log(1+Z_k(q))\right)^{\alpha}}{Z_k(q)}|Z'_k(q)|\right] \leq \frac{2}{\alpha}+\frac{2}{1+\alpha}(\log\,Z_k^*+1)^{1+\alpha},
$$
which is at most $2(1+e)/\alpha$ for $\alpha<1/\log\,Z_k^*$.

Now consider the term
$$
\sup\limits_q\left|\frac{Z_k(q)}{\left(\log(1+Z_k(q))\right)^{\alpha}}(\hat{b}(q)-b(q))\right|
$$
Note that $\log(1+t)\ge \min\{1,t\}/2$. So the first term can be bounded from above as
$$
\frac{Z_k(q)}{\left(\log(1+Z_k(q))\right)^{\alpha}} \leq 2^{\alpha} \max \left\{
Z_k(q),\,(Z_k(q))^{1-\alpha}
\right\}
$$
The second term behaves as
$$
(\hat{b}(q)-b(q))=-x'(q)v(q)(\hat{G}(b(q))-G(b(q)))
$$
where $G$ and $\hat{G}$ are the real and empirical bid distributions, respectively.
We recall that $\expect{\sup_q\left|\hat{G}(b(q))-G(b(q))\right|} \leq \frac{1}{2\sqrt{N}}$.
Thus
\begin{align*}
\expect{\sup\limits_q\left|\frac{Z_k(q)}{\left(\log(1+Z_k(q))\right)^{\alpha}}(\hat{b}(q)-b(q))\right|}
& \leq 2^{\alpha}\sup_q\left(\max \left\{
x'_k(q),\,(x'_k(q))^{1-\alpha}(x'(q))^{\alpha}
\right\}
\right)\frac{1}{2\sqrt{N}}\\
& \le
2^{\alpha}\sup_q\left(x'_k(q)\right)\left(\underbrace{\max\left(1,\sup_{q:x'_k(q)\ge
      1}\frac{x'(q)}{x'_k(q)}\right)}_{=: \, A}\right)^{\alpha} \frac{1}{2\sqrt{N}}
\end{align*}


Now we combine the two evaluations together and pick
$\alpha=\min\{1/\log A, 1/\log Z_k^*\}$, with $A$ defined as above, to obtain
\begin{align*}
\expect{|\hat{P}_k-P_k|} &\leq 
\frac{2(1+e)}{\alpha}2^{\alpha}A^{\alpha}\frac{1}{2\sqrt{N}} \sup_q\left(x'_k(q)\right)\\
& \leq \frac{20}{\sqrt{\samples}} \,\,
\sup_{\quant}\{\kalloc'(\quant)\}\,\,\max\left\{\log A,\log \sup_{\quant}
\left\{ \frac{\kalloc'(\quant)}{\alloc'(\quant)} \right\}\right\}
\end{align*}
\end{proofof}

\begin{proofof}{Lemma~\ref{lem:bias-bound}}
We can write the function $\tilde{b}(i/N)$ as
\begin{align*}
\tilde{b}\left(\frac{i}{N}\right)& =\frac{N!}{(i-1)!(N-i)!}\int G(t)^{i-1}(1-G(t))^{N-i}g(t) t\,dt\\
&=\frac{N!}{(i-1)!(N-i)!}\int t^{i-1}(1-t)^{N-i}b(t)\,dt
\end{align*}
Note that 
$$\frac{N!}{(i-1)!(N-i)!} t^{i-1}(1-t)^{N-i}$$
is the density of the beta-distribution with parameters $\alpha=i$ and $\beta=N-i+1$. Denote this density $f(t;\alpha,\beta)$. Then we can write
$$  \tilde{b}\left(\frac{i}{N}\right)=\int^1_0b(t)f(t;\alpha,\beta)\,dt.$$
Now let $q \in [i/N,\,(i+1)/N]$, and consider an expansion of $b(t)$ at $q$ such that 
$$ b(t)=b(q)+b'(q)(t-q)+O((t-q)^2)$$
Now we substitute this expansion into the formula for $\tilde b(\cdot)$ above to get
$$  \tilde{b}\left(\frac{i}{N}\right)=b(q)+b'(q)\int^1_0(t-q)f(t;\alpha,\beta)\,dt+O(\int^1_0(t-q)^2f(t;\alpha,\beta)\,dt)$$
The mean of the beta distribution is $\alpha/(\alpha+\beta)$ and the variance is $\alpha\beta/((\alpha+\beta)^2(\alpha+\beta+1))$. This means that 
$$ \tilde{b}\left(\frac{i}{N}\right)-b(q)=b'(q)\left(\frac{i}{N+1}-q\right)+O\left(\frac{1}{N^2}\right). $$
Thus
$$\sup_{q \in [i/N,(i+1)/N]}\left| \tilde{b}\left(\frac{i}{N}\right)-b(q)\right| \leq \sup_qb'(q)\frac{2}{N}+O\left(\frac{1}{N^2}\right).$$
Therefore, the expectation $\left|\Phat_y-\expect{\Phat_y}\right|$ is at most $O(1)/N
\, \sup_q\{\alloc'(q)\} \, \sup_q Z_y(q)$.
\end{proofof}

\begin{proofof}{Lemma~\ref{lem:independent-sums}}
  Fix $I$ and $\ell$, and note that the function $\btild$ is fixed
  (that is, it does not depend on the empirical bid sample). Then, the
  sum $\T_{i_{\ell}, i_{\ell+1}}$ depends only on the empirical bid
  values $\bwhat(\quant)$ for quantiles in the interval $[i_\ell/N,
  i_{\ell+1}/N)$. By the definition of $I$, we know that the smallest
  $i_\ell$ bids in the sample are all smaller than
  $\btild((i_\ell-1)/N)\le\btild(i_\ell/N)$, and the largest $N-i_{\ell+1}$ bids in the
  sample are all larger than $\btild(i_{\ell+1}/N)\ge \btild((i_{\ell+1}-1)/N)$. On the other
  hand, the empirical bids $\bwhat(\quant)$ for $\quant\in [i_\ell/N,
  i_{\ell+1}/N)$ lie within $[\btild(i_\ell/N),
  \btild((i_{\ell+1}-1)/N)]$. Therefore, conditioned on $i_\ell$ and
  $i_{\ell+1}$, the latter set of empirical bids is independent of the
  former set of empirical bids.
\end{proofof}

\begin{proofof}{Lemma~\ref{lem:deviation-prob} and Theorem~\ref{thm:inference-y}}
  We will use Chernoff-Hoeffding bounds to bound the expectation of $\T_{0,N}$ over the bid sample, conditioned on $I$ and $\Delta = \sup_{\quant\in [0,1]}|\hat{G}(\bid(\quant))-G(\bid(\quant))|$. We first note that $\T_{0,N}$ has mean zero because for any integer $i\in [0,N]$, $\expect[\text{samples}]{\bwhat(i/N)} = \btild(i/N)$.

Next we note that the $\Tij$'s are bounded random variables. Specifically, let $Q$ be an interval of quantiles over which the difference $\bwhat(\quant)-\btild(\quant)$ does not change sign. Then, following the proof of Theorem~\ref{th: all pay}, we can bound 
\begin{align*}
|\T_Q| & = \left|\int_Q Z'_y(\quant)(\bwhat(\quant)-\btild(\quant)) \, \text{d}q \right|\\
& \le 20\Delta \underbrace{\,\,\sup_{\quant}\{y'(\quant)\}\,\,\log \max \left\{ \sup_{\quant: y'(\quant)\ge 1} \frac{\alloc'(\quant)}{y'(\quant)},
\sup_{\quant}\frac{y'(\quant)}{\alloc'(\quant)} \right\}}_{=:\, S}.
\end{align*}
Likewise, over an interval $Q$ where $Z'_y$ does not change sign, we again get $|\T_Q|\le 20 \Delta S$ with $S$ defined as above. Moreover, for an interval $Q$ over which $Z'_y$ changes sign at most $t$ times, we have 
$$ \int_Q |Z'_y(\quant)(\bwhat(\quant)-\btild(\quant))| \, \text{d}q \le t\cdot 20\Delta S.$$ Finally, noting that $Z_y$ is a weighted sum over the $n$ functions $Z_k$ defined for the $k$-unit auctions, and that by Lemma~\ref{lem:Z-bound-1} each $Z_k$ has a unique maximum, we note that $Z'_y$ changes sign at most $2n$ times.

We now apply Chernoff-Hoeffding bounds to bound the probability that the sum $\sum_{\ell=0}^{\ell=k-1}\T_{i_{\ell}, i_{\ell+1}}$ exceeds some constant $a$. With $\tau_\ell$ denoting the upper bound on $|\T_{i_{\ell}, i_{\ell+1}}|$, this probability is at most 
$$\text{exp}\left(-\frac{a^2}{\sum_{\ell} \tau_\ell^2}\right).$$
By our observations above, for all $\ell$, $\tau_\ell\le 40\Delta S$, and $\sum_\ell \tau_\ell \le \int_0^1 |Z'_y(\quant)(\bwhat(\quant)-\btild(\quant))| \, \text{d}q \le 40n\Delta S$. Therefore, $\sum_{\ell} \tau_\ell^2\le n(40\Delta S)^2$. We can now choose $a = \sqrt{n\log n}\, 40 \Delta S$ to make the above probability at most $1/n$.

Putting everything together, we get that conditioned on $I$ and $\Delta$, the expected value of $|\T_{0,N}|$ over the bid sample is at most $a + 1/n \cdot 40n\Delta S = O(1)\sqrt{n\log n}\Delta S$. Since this bound is independent of $I$, the same bound holds when we remove the conditioning on $I$. The theorem now follows by plugging in the expected value of $\Delta$ from Lemma~\ref{error bid function}.
\end{proofof}

\begin{proofof}{Theorem~\ref{theorem:sign-AB}}
We need to bound the probability that the error in estimating
$\hat{P}_{B_1}-\alpha\hat{P}_{B_2}$ is greater than $|P_{B_1}-\alpha\,P_{B_2}|$. This
error can in turn be decomposed into the error in estimating $P_{B_1}$ and
that in estimating $P_{B_2}$. Denote $a=|P_{B_1}-\alpha\,P_{B_2}|>0$. Then,
\begin{align*}
\Pr\left(|(\widehat{P}_{B_1}-\alpha\,\widehat{P}_{B_2})-(P_{B_1}-\alpha\,P_{B_2})|>a\right) \leq
\Pr\left(|\widehat{P}_{B_1}-P_{B_1}|>a/2\right)
+ \Pr\left(|\widehat{P}_{B_2}-P_{B_2}|>a/2\alpha\right).
\end{align*}
Let $\alloc$ denote the allocation rule of the mechanism $C$ that we
are running, and let $\bid$ be the corresponding bid function. Now, recall that for
$$\Delta=\sup_q|\widehat{G}(b(q))-G(b(q))|,\,\,\, \text{and }\,\,
S=\sup_{\quant}\{x_{B_1}'(\quant)\}\log\max\left\{\sup_{\quant: x'_{B_1}(\quant)\ge 1}\frac{\alloc'(\quant)}{x_{B_1}'(\quant)},
\sup_{\quant}\frac{x_{B_1}'(\quant)}{\alloc'(\quant)} \right\},$$ 
Equation~\eqref{eq:error-allpay-y-hat},
Lemma~\ref{lem:bias-bound}, and Lemma~\ref{lem:deviation-prob} together imply that (conditional on $\Delta$)
$$
\Pr\left(|\widehat{P}_{B_1}-P_{B_1}|>a/2\right) \leq
2\exp\left(-\frac{1}{n(40\,\Delta\,S)^2}\left(\frac a2-O\left(\frac n{\epsilon N}\right)\right)^2\right).
$$
Now recall that $S<n\log (n/\eps)$, and
$\Delta<\text{constant}/\sqrt{N}$ with high probability
(Lemma~\ref{error bid function}). As a result, we establish that for
$a=\omega(n/\eps N)$,
$$
\Pr\left(|\widehat{P}_{B_1}-P_{B_1}|>a/2\right) \leq
\exp\left(-O\left(\frac{Na^2}{n^3\log(n/\eps)}\right)\right).
$$ 
Likewise,
$$
\Pr\left(|\widehat{P}_{B_2}-P_{B_2}|>a/2\alpha\right) \leq
\exp\left(-O\left(\frac{Na^2}{\alpha^2n^3\log(n/\eps)}\right)\right).
$$
\end{proofof}

\end{document}